\documentclass[journal,onefignum,onetabnum]{siamart171218}
\usepackage{cases}
\usepackage{bbm}
\usepackage{float}
\usepackage{subfig}
\usepackage{comment}

\usepackage{lipsum}
\usepackage{amsfonts}
\usepackage{graphicx}
\usepackage{epstopdf}
\usepackage{algorithmic}
\ifpdf
  \DeclareGraphicsExtensions{.eps,.pdf,.png,.jpg}
\else
  \DeclareGraphicsExtensions{.eps}
\fi

\newsiamremark{remark}{Remark}
\newsiamremark{hypothesis}{Hypothesis}
\crefname{hypothesis}{Hypothesis}{Hypotheses}
\newsiamthm{claim}{Claim}
\newsiamthm{assumption}{Assumption}

\headers{Nonlinear social contagion dynamics}{Xunlong Wang, Feng Fu, and Bin Wu}

\title{NONLINEAR SOCIAL CONTAGION DYNAMICS ON DYNAMICAL NETWORKS: EXACT SOLUTIONS FOR CONSENSUS TIMES AND EVOLUTIONARY TRAJECTORIES\thanks{Submitted to the editors DATE.
\funding{This work is supported by National Natural Science Foundation of China (NSFC) under Grant No. 61751301.}}}

\author{Xunlong Wang\thanks{School of Science, Beijing University of Posts and Telecommunications, Beijing 100876, China.
\and Key Laboratory of Mathematics and Information Networks (Beijing University of Posts and Telecommunications), Ministry of Education, China (\email{bin.wu@bupt.edu.cn}).}
\and Feng Fu\thanks{Department of Mathematics, Dartmouth College, Hanover, 03755, NH, USA 
  (\email{feng.fu@dartmouth.edu}).}
\and Bin Wu\footnotemark[2]}

\usepackage{amsopn}

\makeatletter
\newcommand*{\addFileDependency}[1]{
  \typeout{(#1)}
  \@addtofilelist{#1}
  \IfFileExists{#1}{}{\typeout{No file #1.}}
}
\makeatother

\newcommand*{\myexternaldocument}[1]{%
    \externaldocument{#1}%
    \addFileDependency{#1.tex}%
    \addFileDependency{#1.aux}%
}

\ifpdf
\hypersetup{
  pdftitle={Nonlinear contagion dynamics on dynamical networks: exact solutions ranging from consensus times to evolutionary trajectories},
  pdfauthor={Xunlong Wang, Feng Fu, Bin Wu}
}
\fi

\myexternaldocument{ex_supplement}

\begin{document}

\maketitle

\begin{abstract}
Understanding nonlinear social contagion dynamics on dynamical networks, such as opinion formation, is crucial for gaining new insights into consensus and polarization. 
Similar to threshold-dependent complex contagions, the nonlinearity in adoption rates poses challenges for mean-field approximations. To address this theoretical gap, we focus on nonlinear binary-opinion dynamics on dynamical networks and analytically derive local configurations, specifically the distribution of opinions within any given focal individual's neighborhood.
This exact local configuration of opinions, combined with network degree distributions, allows us to obtain exact solutions for consensus times and evolutionary trajectories. 
Our counterintuitive results reveal that neither biased assimilation (i.e., nonlinear adoption rates) nor preferences in local network rewiring -- such as in-group bias (preferring like-minded individuals) and the Matthew effect (preferring social hubs) -- can significantly slow down consensus. 
Among these three social factors, we find that biased assimilation is the most influential in accelerating consensus. 
Furthermore, our analytical method efficiently and precisely predicts the evolutionary trajectories of adoption curves arising from nonlinear contagion dynamics.
Our work paves the way for enabling analytical predictions for general nonlinear contagion dynamics beyond opinion formation.
\end{abstract}



\section{Introduction}
Contagion dynamics plays an important role in all aspects of life, ranging from the spreading of infectious diseases \cite{Mason2019NonlinearEpidemic,PRL2006_epidemic_ThiloGross, PRE2022Liuyuan}, the dissemination of rumors \cite{Nature1964Rumour}, the formation of public opinions \cite{PRL2008Whotalking, Socialmedia_empirical} to 
the polarization of elected politicians \cite{PRX_WangXin2020, SIAM2024Mason}.
Statistical methods provide plenty of tools to study contagion dynamics on networks, such as mean-field equations \cite{Mean-fieldBOOK, Gleeson2012Meanfield}.

Exact analytical results for contagion dynamics on networks are challenging.
On one hand, individuals have local network configurations \cite{ER1959, Science1999BAmodel}, which can differ from each other.
This diversity leads to high-dimensional systems, falling into the curse of dimension.
On the other hand, the complexity is enhanced by the non-linearity of spreading rates. Taking epidemic dynamics as an example, this non-linearity refers to that the probability of being infected is a nonlinear function of the number of infected neighbors.

The non-linearity is ubiquitous and non-neglectable.
This happens if individuals are assumed to adopt opinions via the majority principle, i.e., the opinion is adopted as long as it occupies more than half of the neighborhood \cite{PRE2009Qvoter, Liggett1985Major}.
This can also happen if the susceptible spontaneously avoid contact with others for the sake of health, which gives rise to the nonlinear spreading rate in epidemic dynamics \cite{epidemic1986nonlinear,epidemicChaos2021nonlinear}.
Homogeneous pair approximations provide accurate predictions for contagion dynamics with linear spreading rates \cite{NJP2008analytical,Gleeson2012Meanfield,PRE2024QvoterHP,PlosOne2010Bin}.
However, all of these methods can lead to inaccurate and sometimes even wrong predictions for nonlinear contagion dynamics on networks \cite{PRX_2013_SKM,EPL2009HPA, NJP2014HPA,PRE2020PA}.
An alternative approach for nonlinear contagion dynamics is the approximated master equations, which fixes the inaccuracy of pair approximation \cite{PRX_2013_SKM, PRR2020SKM, NSR2017SKM, PRE2013Chenguanrong, SIAM2019Glesson}.
The approximated master equations rely crucially on the likelihood that a susceptible/infected individual has a given number of susceptible neighbors and another given number of infected neighbors.
This likelihood is the full information of almost all the contagion dynamics on networks \cite{PNAS_2012_SKM, Chaos_2016_SKM, PRL_2011_SKM} (except for game interactions \cite{Nature2006SimpleRule}).
Up till now, this likelihood has only been numerically estimated with no analytical results. 
This quantity is even more challenging if the dynamical nature of networks is taken into account.

In this paper, we obtain analytically the likelihood on dynamical networks,
where heterogeneous duration times of the social relationships are taken into account together with the Mathew effect (preferential attachment).
We show the accuracy of the likelihood via a rigors proof.
As an application, we show how the likelihood sheds light on the absorbing time for nonlinear contagion dynamics,
which is crucial in understanding how diseases, behaviors, or information spreads in structured populations \cite{Interface2024Gaoshun,PhysicaA2019clique,Sui2015PREspeed}.

\section{Model}
We propose a binary-state model where contagion dynamics and linking dynamics co-evolve. 
Opinions propagate with probability $\phi$ while the network evolves with probability $1-\phi$.
The network consists of $N$ nodes and $L$ links, and the average degree of the network is $\Bar{k}=2L/N$.
Nodes in the network represent individuals in social networks and links represent social ties.
Each node in the network has an opinion: either $A$ or $B$.
Links are thus of two categories: either homogeneous links ($A-A$ and $B-B$) or heterogeneous links ($A-B$).

\begin{figure*}[tb]
    \centering
    \includegraphics[width=0.99\linewidth]{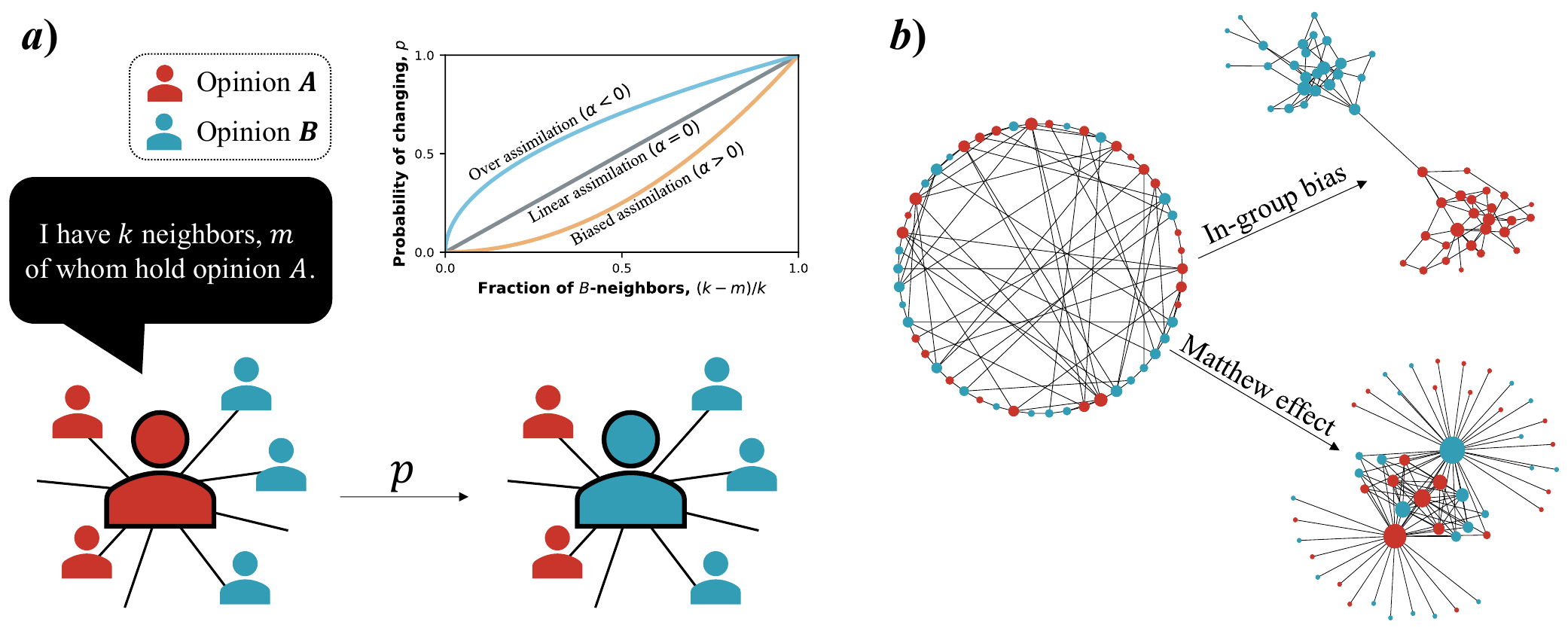}
    \caption{
    Three key features of opinion formation on dynamical networks: (a) the biased assimilation leads to the nonlinear spreading rate $p$, which is related to the opinion, the neighbor size and the number of $A-$neighbors of the focal individual. (b) the in-group bias and the Matthew effect.
    The two opinions are evenly divided and placed randomly in the initial network, whose topology is that $N$ edges join $N$ nodes into a ring and $L-N$ links are placed randomly.
    The in-group bias leads to few interactions between different opinions, leading to the emergent of echo-chamber \cite{PRL2020modelecho}.
    The Matthew effect arising from preferential attachment leads to the emergent of hub nodes, whose degrees are much higher than others, i.e., the opinion leaders.}
    \label{fig: model1}
\end{figure*}

In opinion propagation, a randomly selected individual changes his/her opinion with probability $(m/k)^{1+\alpha}$, if the selected individual has $k$ neighbors, $m$ of whom hold the opposite opinion.
Here $\alpha$ measures the strength of the biased assimilation \cite{PNAS_biasedassimilation,conferencebiased,VVV2019interfaceBiased,VVV2024CSF}.
We assume that individuals have an inclination to disagree when hearing the opposite,
taking into account that individuals prefer their own opinions and undermine opposing opinions \cite{ScienceBiased1974, 1998ConfirmationBiased, Lord1979BiasedAA}.
The assumption is also present in Friedkin-Johnsen model \cite{FJmodel1999} and Hegselmann-Krause model \cite{HKmodel2002}.
If $\alpha=0$, then the opinion propagation degenerates to the voter model \cite{Liggett1985Major}, which follows the linear assimilation as shown in \cref{fig: model1}(a).
Furthermore, non-zero $\alpha$ refers to the nonlinear spreading rates, which has been observed not only in opinion propagation but also in the spreading of disease \cite{NCdeeplearning, NP2019Macroscopic} and the vanishing of languages \cite{Nature2003language, Language2013}.

In the linking dynamics, firstly a link is randomly selected. 
Secondly, the selected link breaks off with probability $k_d$ if it's heterogeneous and with probability $k_s$ otherwise. 
If the selected link doesn't break off, it goes to the first step of the linking dynamics. 
Otherwise, an end of the broken link is chosen randomly and it rewires to an individual $i$ outside the neighborhood with a probability proportional to the power of neighbor size of individual $i$, i.e.,  $k_i^\beta$ with $\beta\ge 0$.
The difference of breaking probabilities $\delta=k_d-k_s\ge 0$ measures the strength of in-group bias \cite{WangXunlongChaos},
i.e., how  individuals prefer to stay away from those with opposite opinions.
For $\delta=0$, the linking dynamics is not adaptive any longer \cite{AdaptiveReview_ThiloGross}.
$\beta\ge 0$ is to measure the strength of the Matthew effect for preferential attachment,
i.e, how individuals are likely to make friends with those with high degrees.
\cite{Matjaz2014Matthew, PhysicaA2015nonlinearattach, Sublinearattachment,superlinearattachment}.
For $\beta=0$, each individual is connected with equal probability provided that the network is not dense $\Bar{k}\ll L$. 
If $\beta>0$, only links whose both ends have more than one neighbor are selected to break off.
It ensures connectivity of the network during the evolutionary process.




\begin{assumption}[Time-scale separation] \label{assumption: time-scale}
    Social relationships evolve much faster than opinions, i.e., $\phi\to 0^+$.
\end{assumption}

The assumption is widely adopted in coevolving network models \cite{PRL2006ArneDynamical, PRE2024WangYakun, PRL2023Liujiazhen, PRL2020modelecho, JTB2022Shanxu, PlosOne2010Bin} and it's of practical significance to opinion propagation in the real world.
In fact, it's hard to change one's attitudes about something important \cite{Timescale1988} in a short time, e.g., faith and religion.

\section{Stationary regime of network topology}

Thanks to \cref{assumption: time-scale}, the network has evolved so many rounds that the network structure has converged to the stationary regime whenever opinions propagate.
In this section, we study the stationary regime of the network topology.

\begin{definition}
    Choose an arbitrary individual with opinion $A$, called Adam. 
     Adam has $X^A(t)$ neighbors, of which $Y^A(t)$ hold opinion $A$ at time $t$.
\end{definition}
For simplicity, it's denoted by $\{X(t), Y(t)\}$. It is a  Markov chain, and captures the exact local configuration in Adam's neighborhood at a given time.
Its  state space is $\mathcal{S}=\left\{(i,j)\in \mathcal{Z}\times\mathcal{Z}|0 \leq i \leq j \leq N-1\right\}$. 
We have $|X(t+1)-X(t)|\le 1$ and $|Y(t+1)-Y(t)|\le 1$, because only one edge can be adjusted at a time. 
\begin{figure}[tbp]
    \centering
    \includegraphics[width=\textwidth]{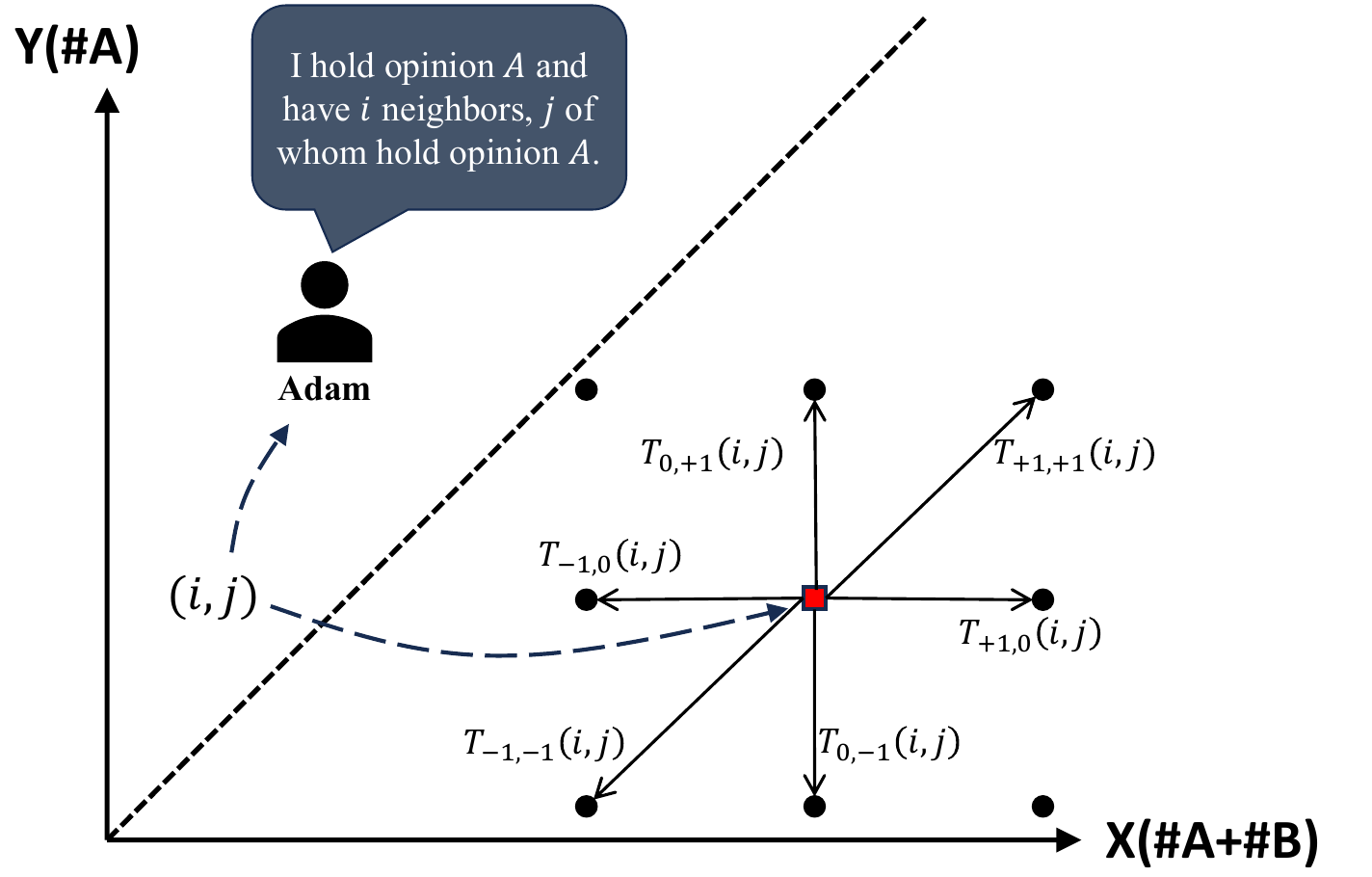}
    \caption{A two-dimensional Markov chain $\{X(t), Y(t)\}$ depicts the evolution of the local network configuration.  
    The coordinate of the red circle is $(i,j)$, which represents Adam has $i$ neighbors consisting of $i-j$ neighbors with opinion $B$ and $j$ neighbors with opinion $A$. 
    Here, the state space is the set of the positive integer coordinates below the dashed line (i.e., the diagonal line).
    Concerning the transition probabilities, it is not likely that Adam has one more $A-$neighbor and simultaneously loses one neighbor in total based on the linking dynamics.
    Thus there are only $6$ transitions from $(i,j)$ to other states.
     For example, the transition $T_{0,+1}(i,j)$ represents the transition from state $(i,j)$ to $(i,j+1)$, which means Adam gains one more neighbor with opinion $A$ and loses a neighbor with opinion $B$. For $\beta=0$, it is found that the underlying Markov chain is ergodic leading to a unique stationary network configuration, which suffices to predict the dynamics of nonlinear opinion formation, ranging from consensus times to evolutionary trajectories. 
   } 
    \label{fig: mp}
\end{figure}
For any given state $(i,j)$ ($0<j<i<N-1$), there are $6$ transitions to neighboring states (as shown in \cref{fig: mp}). 
$T_{n_X,n_Y}(i,j)$ is  denoted as the transition from $(i,j)$ to $(i+n_X,j+n_Y)$, where $X(t+1)-X(t)=n_X$ and $Y(t+1)-Y(t)=n_Y$ ($n_X,n_Y\in\{+1,-1,0\}$).
For example, $T_{+1,+1}(i,j)$ represents the transition from $(i,j)$ to $(i+1,j+1)$, i.e., Adam, who has $j$ neighbors with opinion $A$ and $i-j$ neighbors with opinion $B$, has one more neighbor with opinion $A$ and keeps his neighbors with opinion $B$ unchanged. 
Meanwhile, $P_{n_X,n_Y}(i,j)$ denotes the probability of the transition $T_{n_X,n_Y}(i,j)$, i.e., $P_{n_X,n_Y}(i,j)=\mathbb{P}r[T_{n_X,n_Y}(i,j)]$.
We still take $P_{n_X,n_Y}(i,j)$ as an example.
The transition $T_{+1,+1}(i,j)$ happens in two ways.
One is that, an $A-A$ link that doesn't connect with Adam is selected and breaks off with probability $k_s$, and any end of the broken link connects to Adam.
The other is that an $A-B$ link that doesn't connect with Adam is selected and breaks off with probability $k_d$, and the end with opinion $A$ of the broken link is selected and rewires to Adam.
Analogously, the transition probabilities for $\beta=0$ are given by
\begin{numcases}{}
    P_{+1,+1}(i,j)=\frac{L_{AA}-j}{L}\cdot k_{s}\cdot\frac{1}{N}+\frac{L_{AB}-(i-j)}{L}\cdot k_{d}\cdot\frac{1}{2}\cdot\frac{1}{N},\label{eq: T1}\\
    P_{-1,-1}(i,j)=\frac{j}{L}\cdot k_{s}\cdot\frac{1}{2}\cdot \frac{N-1}{N},\label{eq: T2}\\
    P_{0,+1}(i,j)=\frac{i-j}{L}\cdot k_{d}\cdot\frac{1}{2}\cdot x_A,\label{eq: T3}\\
    P_{0,-1}(i,j)=\frac{j}{L}\cdot k_{s}\cdot \frac{1}{2}\cdot x_B,\label{eq: T4}\\
    P_{+1,0}(i,j)=\frac{L_{AB}-(i-j)}{L}\cdot k_{d}\cdot\frac{1}{2}\cdot\frac{1}{N}+\frac{L_{BB}}{L}\cdot k_{s}\cdot\frac{1}{N},\label{eq: T5}\\
    P_{-1,0}(i,j)=\frac{i-j}{L}\cdot k_{d}\cdot\frac{1}{2}\cdot\frac{N-1}{N}.\label{eq: T6}
\end{numcases}
Here $x_Y$ is the fraction of opinion $Y$ in the whole population (here $x_Y$ remains unchanged based on \cref{assumption: time-scale}) and $L_{XY}$ denotes the number of links of $XY$, where $X, Y\in\{A, B\}$ (the same as below). 

The transition probabilities Eqs.~\eqref{eq: T1}-\eqref{eq: T6} lead to a transition probability matrix $\mathbb{P}$, where each element stores a transition probability and
$\mathbb{P}_{(i,j),(i+n_X,j+n_Y)}=P_{n_X,n_Y}(i,j)$.

\begin{lemma}
    For $\beta=0$, $\{X(t), Y(t)\}$ is homogeneous, irreducible, aperiodic and ergodic.\label{lemma: p4}
\end{lemma}

\begin{proof}
    We will prove these one by one:
    \begin{enumerate}
    \item \textbf{Homogeneous.} The $6$ transitions (as shown in \cref{fig: mp}) of state $(i,j)$ only depend on the current state of itself (i.e., $(i,j)$), but is independent of time, i.e., it's a time-homogeneous Markov process.
    \item \textbf{Irreducible.} 
    For any given two states $(i_1,j_1),(i_2,j_2)\in \mathcal{S}$ (without loss of generality, let $i_1<i_2$ and $j_1>j_2$), if Adam loses $j_1-j_2$ neighbors with opinion $A$ and gains $i_2-j_2-i_1+j_1$ neighbors with opinion $B$, the Markov chain transfer from $(i_1,j_1)$ to $(i_2,j_2)$. 
    Since the probability that Adam gains a neighbor with opinion $B$ or loses a neighbor with opinion $A$ is non-zero, the transition from $(i_1,j_1)$ to $(i_2,j_2)$ can happen in $i_2-2j_2-i_1+2j_1$ steps with a non-zero probability, i.e., $(i_2,j_2)$ is reachable from $(i_1,j_1)$.
    Thus, it's an irreducible Markov chain.
    \item \textbf{Aperiodic.} If a link in the network is selected but fails to break off, Adam's surrounding configuration remains unchanged, i.e., $\{X(t+1),Y(t+1)\}=\{X(t),Y(t)\}$. It can happen with a non-zero probability no matter what the current state is, so the period of each state is $1$.
    Thus it's an aperiodic Markov chain.
    \item \textbf{Ergodic.} Because it's irreducible and aperiodic, it's an ergodic Markov chain.
\end{enumerate}
\end{proof}

A homogeneous, ergodic Markov chain on a finite state space has a unique stationary distribution, which is equal to the limiting distribution \cite{Gardiner1986HandbookOStochastic}.
Based on \cref{lemma: p4}, $\{X(t), Y(t)\}$ on finite state space $\mathcal{S}$ for $\beta=0$ has a unique stationary distribution, which is also the limiting distribution.

\begin{definition}
    Denote the stationary distribution of $\{X(t), Y(t)\}$ by $\pi_{k,m}^A$. \label{def: pi}
\end{definition}

It's usually obtained by solving the left eigenvector of the transition probability matrix $\mathbb{P}$, i.e., $\pi_{k,m}^A=\sum\limits_{\{i,j\}\in\mathcal{S}}\pi_{i,j}^A \mathbb{P}_{\{i,j\},\{k,m\}}$, $\forall\{k,m\}\in\mathcal{S}$. 
However, it's too challenging to solve it because the transition probability matrix $\mathbb{P}$ is too large, which is of $\mathcal{O}(N^2)$. 

\begin{lemma}[Ref.~\cite{PlosOne2010Bin}]\label{lemma: plos}
    For $\beta=0$ and large $N$ (i.e., $N\to\infty$), after a long time the proportions of different kinds of links(i.e., $L_{XY}/L$) converge to a stationary state, i.e., 
    \begin{equation}
        (\frac{L_{AA}}{L},\frac{L_{AB}}{L},\frac{L_{BB}}{L})=(\frac{x_A^2}{k_{s}},\frac{2x_Ax_B}{k_{d}},\frac{x_B^2}{k_{s}})\mathcal{N}(x_A),\label{eq: plos}
    \end{equation}
    where $\mathcal{N}(x_A)=(\frac{x_A^2}{k_{s}}+\frac{2x_Ax_B}{k_{d}}+\frac{x_B^2}{k_{s}})^{-1}$ is a normalization factor. 
\end{lemma}

\begin{proof}
    The detailed proof is shown in Ref.~\cite{PlosOne2010Bin}.
\end{proof}

\begin{theorem}
    For $\beta=0$ and large $N$, the unique stationary distribution of $\{X(t), Y(t)\}$ is
    \begin{equation}
        \pi_{k,m}^A=\frac{\lambda_A^k}{k!}e^{-\lambda_A}\cdot\tbinom{k}{m}p_A^m(1-p_A)^{k-m}, \label{eq: piAkm1}
    \end{equation}
    where $\lambda_A=\frac{2L_{AA}+L_{AB}}{Nx_A}$ and $p_A=\frac{2L_{AA}}{2L_{AA}+L_{AB}}$.
    Besides, it satisfies the detailed balance condition.\label{theorem: piAkm}
\end{theorem}

\begin{proof}
    Large $N$ means the number of $X-$neighbors of Adam is much smaller than the total number of $A-X$ links in the population, i.e., $j\ll L_{AA}$ and $i-j\ll L_{AB}$, where $i$ and $j$ represent the number of Adam's neighbors and Adam's $A-$neighbors.
    And because of \cref{lemma: plos}, the transition probabilities from state $(i,j)$ become 
    \begin{numcases}{}
        P_{+1,+1}(i,j)=\frac{x_A}{N}\mathcal{N}(x_A),\nonumber\\
        P_{-1,-1}(i,j)=\frac{k_{s}}{2L}j,\nonumber\\
        P_{0,+1}(i,j)=\frac{k_{d}\cdot x_B}{2L}(i-j),\nonumber\\
        P_{0,-1}(i,j)=\frac{k_{s}\cdot x_B}{2L}j,\nonumber\\
        P_{+1,0}(i,j)=\frac{x_B}{N}\mathcal{N}(x_A),\nonumber\\
        P_{-1,0}(i,j)=\frac{k_{d}}{2L}(i-j).\nonumber
    \end{numcases}
    Let us denote that $\pi_{k,m}^A=\frac{\lambda_A^k}{k!}e^{-\lambda_A}\cdot\tbinom{k}{m}p_A^m(1-p_A)^{k-m}$. It is shown that
    \begin{numcases}{}
        \pi_{k,m}^A\cdot P_{+1,+1}(k,m)=\pi_{k+1,m+1}^A\cdot P_{-1,-1}(k+1,m+1),\label{eq: db1}\\
        \pi_{k,m}^A\cdot P_{0,+1}(k,m)=\pi_{k,m+1}^A\cdot P_{0,-1}(k,m+1),\label{eq: db2}\\
        \pi_{k,m}^A\cdot P_{+1,0}(k,m)=\pi_{k+1,m}^A\cdot P_{-1,0}(k+1,m),\label{eq: db3}
    \end{numcases}
    hold. In other words, it satisfies the detailed balance condition \cite{Gardiner1986HandbookOStochastic}.
    And the detailed balance is a sufficient condition for the stationary distribution, thus Eq.~\eqref{eq: piAkm1} is the stationary distribution of $\{X(t), Y(t)\}$.
    And because \cref{lemma: p4} leads to the uniqueness of the stationary distribution, Eq.~\eqref{eq: piAkm1} is the unique stationary distribution of $\{X(t), Y(t)\}$, and satisfies the detailed balance condition.
\end{proof}

\cref{theorem: piAkm} leads naturally to the following properties:
\begin{enumerate}
    \item For $\beta=0$ and large $N$, the marginal distribution of the stationary (limiting) distribution is the Poisson distribution, i.e., $\mathbb{P}r[X(\infty)=k]=\lambda_A^ke^{-\lambda_A}/k!$.
    \item For $\beta=0$ and large $N$, if $\{X(t)\}$ is given, then the stationary distribution of $\{Y(t)\}$ is the Binomial distribution, i.e.,
    $\mathbb{P}r[Y(\infty)|X(\infty)=k]=\tbinom{k}{m}p_A^m(1-p_A)^{k-m}$.
    \item For $\beta=0$ and large $N$, the proportion of $A-$individuals who have $m$ neighbors with opinion $A$ and $k-m$ neighbors with opinion $B$ among all $A-$individuals in the networks converges to $\pi_{k,m}^A$.
    \item For $\beta=0$ and large $N$, the proportion of $B-$individuals who have $m$ neighbors with opinion $A$ and $k-m$ neighbors with opinion $B$ among all $B-$individuals in the networks converges to
    \begin{equation}
        \pi_{k,m}^B=\frac{\lambda_B^k}{k!}e^{-\lambda_B}\tbinom{k}{m}p_B^m(1-p_B)^{k-m},\label{eq: piBkm1}
    \end{equation}
    where $\lambda_B=(2L_{BB}+L_{AB})/Nx_B$ and $p_B=L_{AB}/(2L_{BB}+L_{AB})$.\label{theorem: piBkm}
\end{enumerate}

\begin{remark}
    For $\beta=0$ and large $N$, the stationary distribution of $\{Y(t)\}$ conditional on that $\{X(t)\}$ is given is the stationary distribution of $\{Y(t)|X(t)\equiv k\}$, i.e.,
    $\mathbb{P}r[Y(\infty)|X(\infty)=k]=\mathbb{P}r[Y(\infty)|X(\infty)\equiv k]$.\label{remark: binomial}
\end{remark}

\begin{proof}
    Let us focus on $\{Y(t)|X(t)\equiv k\}$, which is a homogeneous, ergodic, one-dimensional Markov chain on a finite state space $\{0,1,\dots, k\}$, so it has a unique stationary distribution.
    Noting that the Markov chain only transfers to the neighboring states or stays unchanged in one step, its stationary distribution satisfies the detailed balance condition, i.e., Eq.~\eqref{eq: db2}.
    It implies the solution of Eq.~\eqref{eq: db2} is unique.
    And because $\pi_{k,m}^A$ also satisfies Eq.~\eqref{eq: db2} no matter what $k$ is, the stationary distribution of $\{Y(t)|X(t)\equiv k\}$ is equal to the stationary distribution $\pi[Y(t)|X(t)=k]$.

    It's easy to prove that the stationary (limiting) distribution of $\{Y(t)|X(t)\equiv k\}$ is the Binomial distribution, i.e.,
    \begin{equation}
        \mathbb{P}r[Y(\infty)=m|X(t)\equiv k]=\tbinom{k}{m}p_A^m(1-p_A)^{k-m}.\nonumber
    \end{equation}
    Thus $\mathbb{P}r[Y(\infty)|X(\infty)=k]=\tbinom{k}{m}p_A^m(1-p_A)^{k-m}$ is proven in another way.
\end{proof}

The stationary regime $\pi_{k,m}^X$ is the product of the degree distribution and opinion distribution in the neighborhood with a given neighbor size (see \cref{fig: topology1}).
On one hand, the degree distribution of the stationary regime follows the Poisson distribution whose average is $\lambda_X$.
The result not only covers the average degree of $X-$individuals, $\lambda_X$, but also provides the exact degree distribution.
On the other hand, opinion distribution follows the Binomial distribution.
This indicates for $\beta=0$ (i.e., rewire-to-random), the opinions of any two neighbors of a given individual are independent.
This is because the Binomial distribution is essentially the sum of independent Bernoulli random variables.
It also implies that the opinion of any neighbor of a given individual and its neighbor size are independent.
For in-group bias (i.e., $\delta=k_d-k_s>0$), then $p_A>x_A>p_B$, i.e., the clustering of individuals with the same opinion.
For large $\delta$, the links between individuals with different opinions are very few while most links connect nodes with the same opinion as shown in \cref{fig: model1}(b), i.e., mirroring the echo chamber effect in sociology \cite{PRX_WangXin2020, PNASnexus2024FuFeng, PRL2020modelecho, PNAS2021Echosocialmedia}.

\begin{figure*}[tb]
    \centering
    \includegraphics[width=0.99\linewidth]{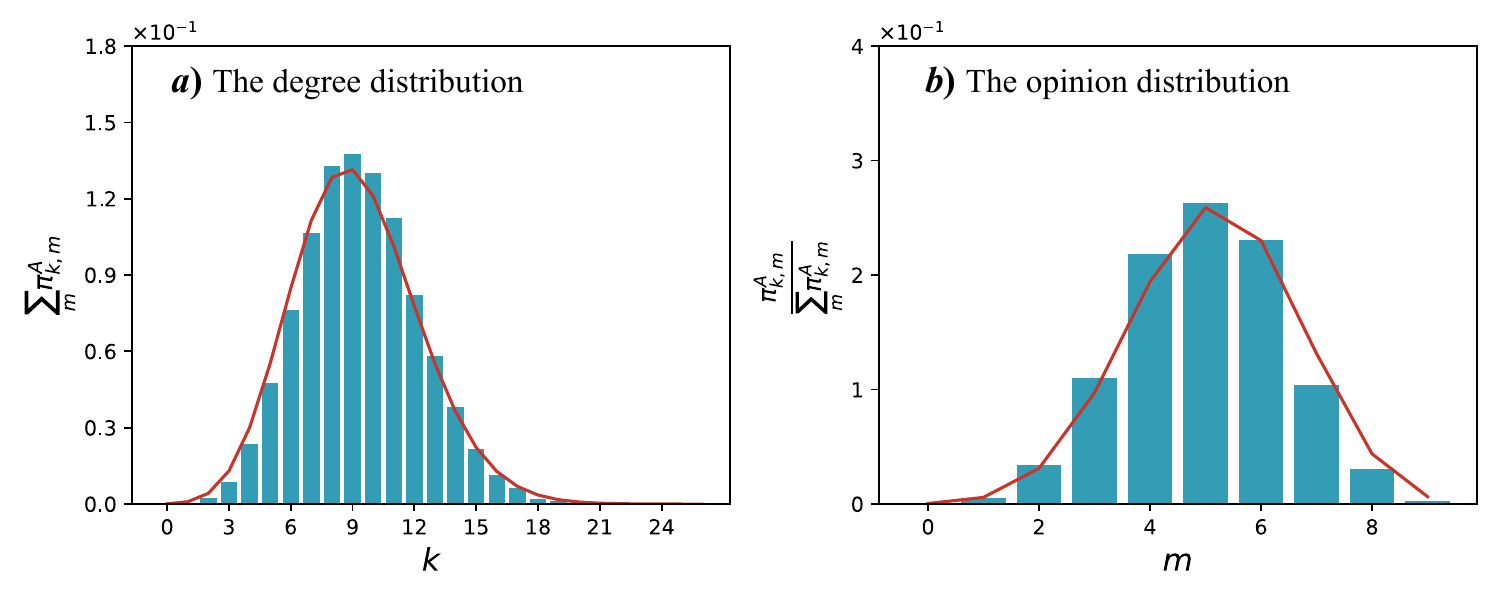}
    \caption{The theoretical $\pi_{k,m}^A$ precisely captures the exact stationary regime of the local network configuration. It can be divided into two parts: (a) the degree distribution of opinion $A$ and (b) the opinion distribution in an $A-$individual's neighborhood with a given size (here $k=9$).
    The solid lines are the theoretical solutions: (a) the Poisson distribution and (b) the Binomial distribution.
    The bars are the simulation results by averaging $100$ independent runs. 
    Here $N=100$, $L=500$, $x_A=0.4$, $\beta=0$, $k_d=1$ and $k_s=0.5$.}
    \label{fig: topology1}
\end{figure*}

\section{Consensus time}

\begin{definition}
    Denote $t_{i}(\alpha^{\prime},\delta^{\prime},\beta^{\prime})$ that the time to reach consensus starting from state in which there are $i$ individuals with opinion $A$ under $\alpha=\alpha^{\prime}$, $\delta=\delta^{\prime}$ and $\beta=\beta^{\prime}$.
    Here an opinion updating is taken as one step.
\end{definition}

To show how powerful the local network configuration $\pi_{k,m}^A$ on nonlinear contagion dynamics is,
we focus on the consensus time (absorbing time) $t_{N/2}$, i.e., how long the population reaches consensus starting from half-and-half.
It fosters the understanding how quickly a group reaches agreement, and provides insights into  decision-making processes.
We try to investigate how $t_{N/2}(\alpha,\delta,\beta)$ is influenced by the biased assimilation $\alpha$, the in-group bias $\delta$ and the Matthew effect $\beta$.
In particular, we try to figure out which of the three affects consensus time most strongly.

To this end, the  consensus time is expanded around $(0,0,0)$, i.e.,  
$t_{N/2}(0,0,0)\!+\!\sum\limits_{\sigma\in
    \{ \alpha,\delta,\beta\}}\sigma\cdot\partial_\gamma t_{N/2}(0,0,0)
    \!+\! o(\sqrt{\alpha^2\!+\!\delta^2\!+\!\beta^2}).
$
Here $\partial_\sigma t_{N/2}(0,0,0)$ measures the impact of weak biased assimilation \cite{VVV2019interfaceBiased}, weak in-group bias \cite{WangXunlongChaos, PRE2016Collective} and weak Matthew effect \cite{PRL2005VoterHeterogeneous, PRE2008VoterHeterogeneous} on consensus time.
The sign represents whether it speeds up consensus.
The magnitude represents the intensity of the impact.

\begin{definition}
    Denote by $\{Z(t)\}$ that there are $Z$ individuals with opinion $A$ at time $t$.
    A single time increment represents the occurrence of an opinion propagation.
\end{definition}

It's a homogeneous Markov chain in a finite state space $\mathcal{E}=\{0,1,2,\dots, N\}$.
For $\beta=0$, the transition probability from state $j$ to $j+1$ is exactly given by
\begin{align}
    T_j^+\!=\!x_B\!\cdot\!\sum\limits_{k=0}^{N-1} \sum\limits_{m=0}^{k}\pi^B_{k,m}\!\cdot\!\left(\frac{m}{k}\right)^{1+\alpha},j=1,2,3,\dots\label{eq: noT1+}
\end{align}
Similarly, the transition probability from state $j$ to $j-1$ is given by
\begin{align}
    T_j^-\!=\! x_A\!\cdot\!\sum\limits_{k=0}^{\infty} \sum\limits_{m=0}^{k}\pi^A_{k,m}\!\cdot\!\left(\frac{k-m}{k}\right)^{1+\alpha},j=1,2,3,\dots\label{eq: noT1-}
\end{align}
Here $x_A=j/N$ and $x_B=(N-j)/N$.
The sums over $k$ from $1$ to $N-1$ in Eqs.~\cref{eq: noT1+,eq: noT1-} are approximated by the sums from $0$ to $\infty$, because $N$ is sufficiently large and the fraction of isolated nodes $e^{-\Bar{k}}$ is relatively small.

The state $0\in\mathcal{E}$ denotes that there aren't any individuals holding opinion $A$ in the population, i.e., the consensus of opinion $B$. 
It's an absorbing state where $\mathbb{P}r[0\to i]=\mathbbm{1}_{\{i=0\}}$. Once the Markov chain $\{Z(t)\}$ reaches the absorbing state, it will no longer move to other states.
Similarly, the state $N\in\mathcal{E}$ is an absorbing state too, which denotes the consensus formation of opinion $A$. Thus the Markov chain $\{Z(t)\}$ is a reducible Markov chain and it reaches one of the two absorbing states after a sufficiently long time, i.e., $Z(\infty)\in\{0, N\}$. 
Thus the consensus time is the absorbing time of underlying Markov chain $\{Z(t)\}$, i.e., $t_j=\min \{t:Z(t)=0\vee N, Z(0)=j\}$.

\begin{proposition}[Consensus time]\label{prop: time}
    If $Z(0)=N/2$, the consensus (absorbing) time $t_{N/2}$ is equivalently given by the following three forms,
\begin{equation}
    -t_1\cdot\sum\limits_{k=N/2}^{N-1}\prod\limits_{m=1}^k\gamma_m+\sum\limits_{k=N/2}^{N-1}\sum\limits_{l=1}^k\frac{1}{T_l^+}\prod\limits_{m=l+1}^k\gamma_m\label{eq: t11}
\end{equation}
\begin{equation}
        \sum\limits_{k=N/2}^{N-1}\sum\limits_{l=1}^k\frac{1}{T_l^+}\prod\limits_{m=l+1}^k\gamma_m-\frac{1}{2}\sum\limits_{k=1}^{N-1}\sum\limits_{l=1}^k\frac{1}{T_l^+}\prod\limits_{m=l+1}^k\gamma_m\label{eq: t22}
    \end{equation}
    \begin{equation}
        \frac{1}{2}\!\left[\sum\limits_{l=1}^{N/2}\frac{1}{T_l^+}\left(\sum\limits_{k=N/2}^{N-1}\prod\limits_{m=l+1}^k \gamma_m\right)\!+\!\sum\limits_{l=N/2+1}^{N-1}\frac{1}{T_l^+}\left(\sum\limits_{k=l}^{N-1}\prod\limits_{m=l+1}^k \gamma_m\right)\!-\!\sum\limits_{l=1}^{N/2}\frac{1}{T_l^+}\left(\sum\limits_{k=l}^{N/2}\prod\limits_{m=l+1}^k \gamma_m\right)\!\right]\label{eq: t33}
    \end{equation}
    where $\gamma_j=T_j^-/T_j^+$ and $t_1=\left(1+\sum\limits_{k=1}^{N-1}\prod\limits_{j=1}^k\gamma_j\right)^{-1}\sum\limits_{k=1}^{N-1}\sum\limits_{l=1}^k\frac{1}{T_l^+}\prod\limits_{j=l+1}^k\gamma_j$.
\end{proposition}

\begin{proof}
    The detailed proof is shown in \cref{appendix: proof time}.
\end{proof}

We study the consensus time $t_{N/2}(0,0,0)$ of the neutral model.
Based on Eq.~\eqref{eq: noT1+}, the transition probability $T_j^+$ of $Z(t)$ is given by 
\begin{align}
    T_j^+=&x_B\cdot\sum\limits_{k=0}^{\infty} \frac{\lambda_B^k}{k!}e^{-\lambda_B}\sum\limits_{m=0}^{k}\tbinom{k}{m}p_B^m(1-p_B)^{k-m}\cdot\left(\frac{m}{k}\right),\label{eq: n_tran}
\end{align}
where $\lambda_B=\Bar{k}$ and $p_B=x_A$.
In order to simplify the form, we define an operator as followed.
\begin{definition}
    Define an operator $\mathbb{E}_{Y}[\mathrm{f}(X)]$, which represents $\mathbb{E}[\mathrm{f}(X)]$ if random variable $X$ follows the distribution $Y$.\label{def: E}
\end{definition}

For example, $\mathbb{E}_{B(k,x_A)}[X]$ represents the first order moment (i.e., the expectation) of the Binomial distribution $B(k,x_A)$, i.e., $\mathbb{E}_{B(k,x_A)}[X]=\sum\limits_{x=0}^{k}\tbinom{k}{x}x_A^x(1-x_A)^{k-x}\cdot x=kx_A$. Based on \cref{def: E}, the transition probability $T_j^+$ is also given by
\begin{align}
    T_j^+=&x_B\cdot\sum\limits_{k=0}^{\infty} \frac{\Bar{k}^k}{k!}e^{-\Bar{k}}\mathbb{E}_{B(k,x_A)}(X)\cdot\frac{1}{k}\nonumber\\
    =&x_B\cdot\sum\limits_{k=0}^{\infty} \frac{\Bar{k}^k}{k!}e^{-\Bar{k}}\cdot x_A\nonumber\\
    =&x_B x_A\cdot\mathbb{E}_{Poisson(\Bar{k})}[1]\nonumber\\
    =&x_Ax_B.\nonumber
\end{align}
Similarly, the transition probability $T_j^-=x_Ax_B$ is obtained.

\begin{proposition}
    For large $N$, the consensus time of neutral model is given by
    \begin{equation}
        t_{N/2}(0,0,0)=\ln{2}\cdot N^2.
    \end{equation}
\end{proposition}

\begin{proof}
    The detailed calculation process is shown in \cref{appendix: tn}.
\end{proof}

The found consensus time on dynamical network agrees with that on complete graphs \cite{PRL2005VoterHeterogeneous}.
Furthermore, our analysis explains why the previous approximations based on mean-field equations are accurate for the linear spreading rate \cite{PlosOne2010Bin, Du2022EvolutionaryDO, Du2024AsymmetricGP, CCC2020} since the exact local network configurations lead to the same transition probabilities for the neutral model.

\subsection{Biased assimilation}

We study $\partial_\alpha t_{N/2}(0,0,0)$.
under weak biased assimilation, the transition probability from state $j$ to $j\pm 1$ is approximated by $T_j^\pm=T_j^\pm\bigg|_{\alpha=0}+\alpha\cdot\partial_\alpha T_j^\pm\bigg|_{\alpha=0}$.
The transition probability from state $j$ to $j+1$ is approximated by $T_j^+=T_j^+\bigg|_{\alpha=0}+\alpha\cdot\partial_\alpha T_j^+\bigg|_{\alpha=0},$
where $\partial_\alpha T_j^+\bigg|_{\alpha=0}\approx x_Ax_B\ln{x_A}+x_B^2\cdot F.$
Here $F=e^{-\Bar{k}}\left[Ei(\Bar{k})-\ln{\Bar{k}}-\gamma\right]/2$ is a constant since the average degree  $\bar{k}$ is invariant during the evolutionary process.
And $Ei(x)$ is the exponential integral function \cite{BOOK2014MathematicalforPhysicists} and $\gamma\approx 0.577$ is the Euler–Mascheroni constant.
The detailed calculation is shown in \cref{appendix: tran_alpha}.
We show that $\partial_\alpha T_j^+\big|_{\alpha=0}$ are good approximations for $\varepsilon <j <N-\varepsilon$ where $\varepsilon=\mathcal{O}(1)\ll N$ is a small positive integer (see \cref{fig: tg}(b)).
In other words, the transition probabilities of $\{Z(t)\}$ in $\{\varepsilon+1,\varepsilon+2,\dots,N-\varepsilon-2,N-\varepsilon-1\}$ can be well approximated.

\begin{figure}[htb]
    \centering
    \subfloat[]{\includegraphics[width=.49\linewidth]{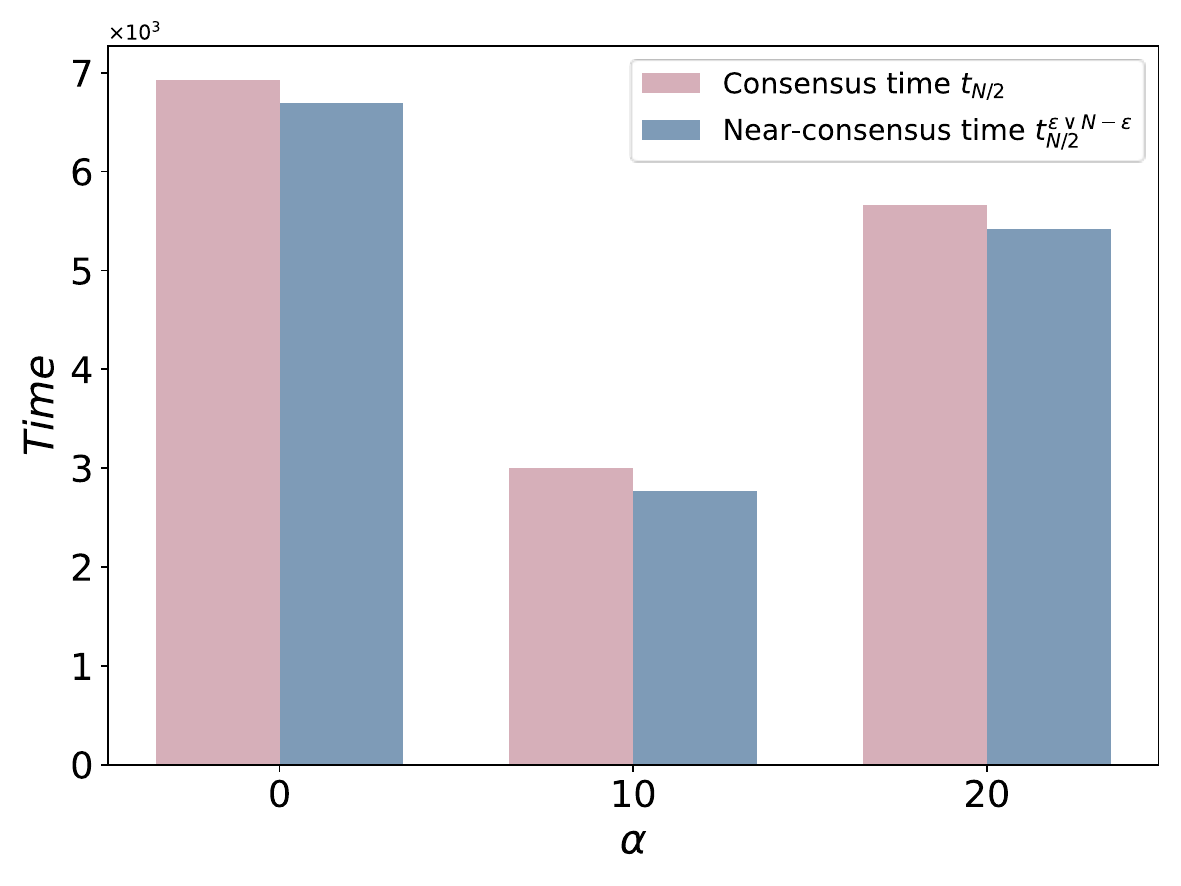}}
    \subfloat[]{\includegraphics[width=.49\linewidth]{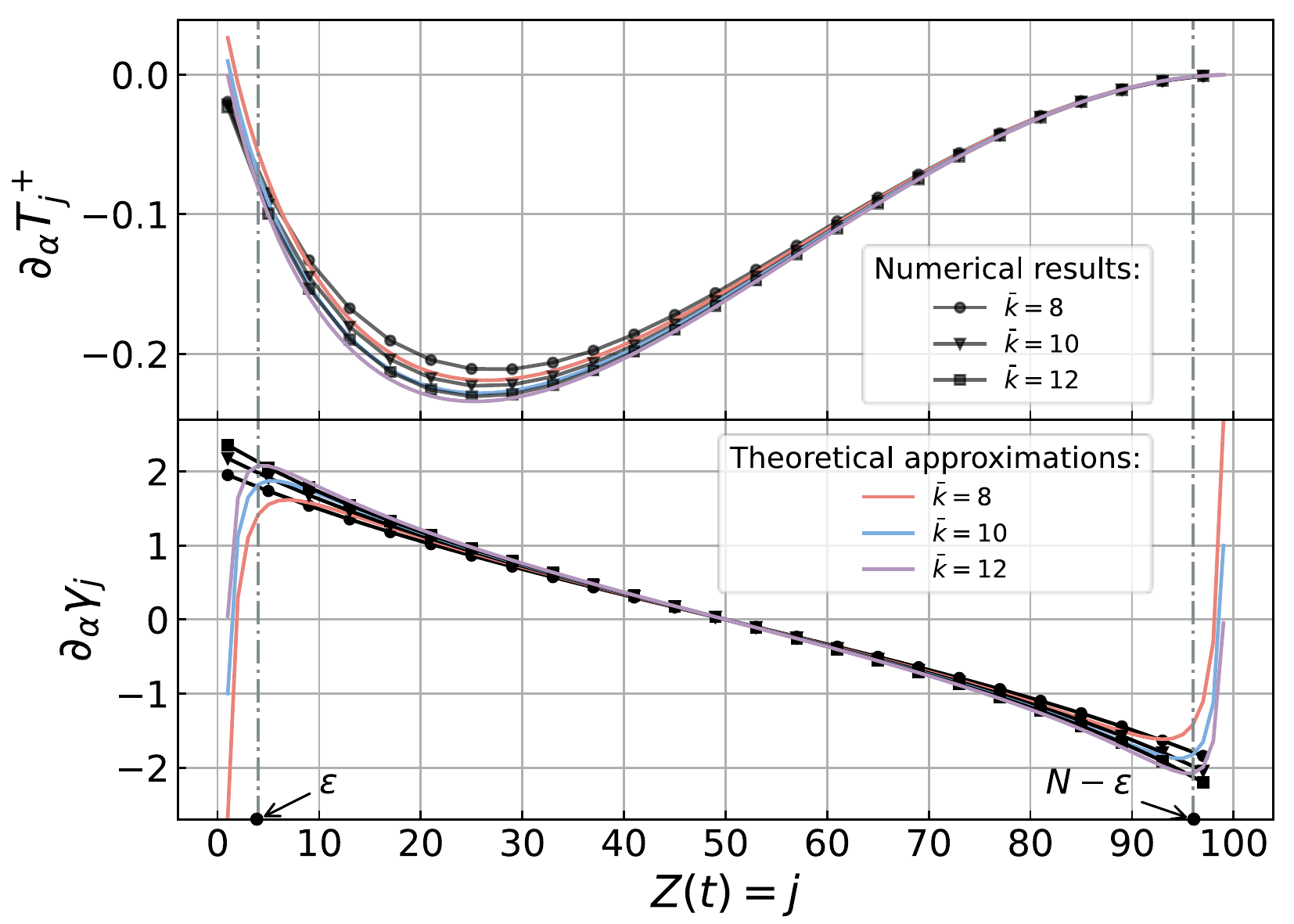}}
    
    \caption{Consensus time $t_{N/2}$ is well approximated by the near-consensus time $t_{N/2}^{\varepsilon\vee N-\varepsilon}$, in which the key intermediate variables $\partial_\alpha T_j^+\big|_{\alpha=0}$ and $\partial_\alpha \gamma_j\big|_{\alpha=0}$ are accurately predicted by Eq.~\eqref{eq: pa1} and Eq.~\eqref{eq: ga1}.
    (a) There is little difference between the consensus time and the near-consensus time for arbitrary $\alpha$.
    Here $N=100$, $L=500$, $\delta=\beta=0$ and $\varepsilon=5$.
    (b) The solid lines of different colors and the dots of different shapes represent the theoretical solutions Eqs.~(\ref{eq: pa1},\ref{eq: ga1}) and the numerical solutions Eqs.~(\ref{eq: pa}, \ref{eq: ga}) by computer. 
    The theoretical solution fits well almost but not well at the ends of the range:
    i) $\partial_\alpha T_j^+\big|_{\alpha=0}<0$ holds based on Eq.~\eqref{eq: pa} but the theoretical solution cannot hold when $j$ is sufficiently small. 
    ii) For $\partial_\alpha \gamma_j\big|_{\alpha=0}$ the Runge-like phenomenon is observed where the theoretical solution coincides well with the numerical solution but there is a pronounced deviation when $j<\varepsilon$ or $j>N-\varepsilon$. 
    It can be observed that $1<\varepsilon\ll N$. 
    Here $N=100$.}
    \label{fig: tg}
\end{figure}

\begin{definition}\label{def: near}
    Define near-consensus time $t_{N/2}^{\varepsilon\vee N-\varepsilon}$, which denotes the first hitting time that it takes to reach state $\varepsilon$ or state $N-\varepsilon$ starting from state $N/2$ for the first time, i.e., $t_{N/2}^{\varepsilon\vee N-\varepsilon}=\inf \{t:Z(t)=\varepsilon\vee N-\varepsilon, Z(0)=N/2\}$.
\end{definition}

The consensus time $t_{N/2}$ is approximated by the near-consensus time $t_{N/2}^{\varepsilon\vee N-\varepsilon}$ thanks to $\varepsilon=\mathcal{O}(1)\ll N$. Thus for large $N$ we obtain
\begin{align}
    \partial_\alpha t_{\frac{N}{2}}(0,0,0)=N^3\ln{N}\cdot\Lambda_\alpha.\label{eq: alpha1}
\end{align}
The detailed calculation is shown in \cref{appendix: t_alpha}.
Noteworthily we show that $\Lambda_\alpha=2F-\ln(\varepsilon+1)/2<0$ by taking into account $\varepsilon\ge 1 $ and $\bar{k}\ge 3.75$.
On one hand, this is because $F$ is a decreasing function of the average degree $\Bar{k}$ and $2F<0.345$ holds when the average degree is not so small, i.e., $\Bar{k}\ge 3.75$.
On the other hand, $\varepsilon\ge 1$ leads to $\ln(\varepsilon+1)/2>0.345$.



\subsection{In-group bias}
We study $\partial_\alpha t_{N/2}(0,0,0)$. 
For $\alpha=\beta=0$, the transition probability from $j$ to $j+1$ is given by
\begin{align}
    T_j^+=&x_B\cdot\sum\limits_{k=0}^{\infty} \frac{\lambda_B^k}{k!}e^{-\lambda_B}\sum\limits_{m=0}^{k}\tbinom{k}{m}p_B^m(1-p_B)^{k-m}\cdot\left(\frac{m}{k}\right)\nonumber\\
    =&x_B\cdot\sum\limits_{k=0}^{\infty} \frac{\lambda_B^k}{k!}e^{-\lambda_B}\mathbb{E}_{B(k,p_B)}[X]\cdot\frac{1}{k}\nonumber\\
    =&x_Bp_B\mathbb{E}_{Poisson(\lambda_B)}[1]\nonumber\\
    =&x_Bp_B.\label{eq: a_tran_delta}
\end{align}
where $\lambda_B$ and $p_B$ is defined in Eq.~\eqref{eq: piBkm1}.
Similarly, for $\alpha=\beta=0$ the transition probability from $j$ to $j-1$ is given by $T_j^-=x_A(1-p_A)$.
The transition probabilities are equal to that by mean-field approximations \cite{PlosOne2010Bin, Du2022EvolutionaryDO}.
This is because the expectation of the Binomial distribution $\mathbb{E}_{B(k,p_X)}(m/k)=p_X$.
It explains why mean-field equations leads to an accurate predictions for linear spreading rates.

Let $k_d=k_0$ and $k_s=k_0-\delta$.
For large $N$ we obtain 
\begin{align}
    \partial_\delta t_{N/2}(0,0,0)=N^3\cdot\Lambda_\delta,\label{eq: delta1}
\end{align}
where $\Lambda_\delta=(\ln{2}-1)/(6k_0)<0$ (see \cref{appendix: t_delta}).

\subsection{Matthew effect}
We study $\partial_\beta t_{N/2}(0,0,0)$. 
From Eq.~\eqref{eq: a_tran_delta}, we note that the transition probabilities of $\{Z(t)\}$ for linear spreading rates depend on the number of individuals with opinion $A$ given its neighbor size (i.e., the Binomial distribution) but not on the degree distribution (i.e., the Poisson distribution).
This turns out to be true even for weak Mathew effect.
In fact, the perturbation of the opinion distribution from the weak Matthew effect is of $\mathcal{O}(\beta^2)$ (see \cref{appendix: perturbation_beta}), leading to
\begin{align}
    \partial_\beta t_{N/2}(0,0,0)=0.\label{eq: beta1}
\end{align}

We find $\partial_\sigma t_{N/2}$ ($\sigma\in\{\alpha,\delta,\beta\}$) are non-positive, which implies either the weak biased assimilation, the weak in-group bias or the weak Matthew effect cannot slow down the formation of consensus.
Furthermore, the speedup from weak biased assimilation is stronger than that of weak in-group bias since $|\partial_\alpha t_{N/2}(0,0,0)|=\mathcal{O}(N^3\ln{N})>|\partial_\delta t_{N/2}(0,0,0)|=\mathcal{O}(N^3)$,
which are validated by simulations (see \cref{fig: weak1}).

\begin{figure}[htb]
    \centering
    \includegraphics[width=\linewidth]{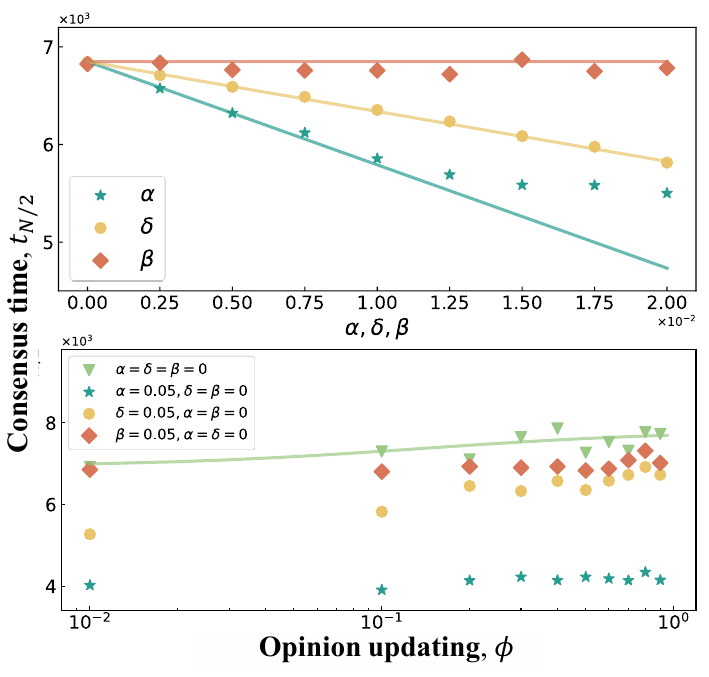}
    \caption{
    The exact $\pi_{k,m}^X$ facilitates us to show that the biased assimilation is the most influential on consensus time among in-group bias, Matthew effect and itself.
    Furthermore, the order is robust, even if we extrapolate from rare opinion updating limit,
    under which our analytical result holds.
    The exact $\pi_{k,m}^X$ precisely predicts the perturbation of transition probabilities by the three effects, leading to accurate predictions of the consensus time.
    Intuitively biased assimilation decreases the probability of opinion shift and in-group bias reduces interactions between different opinions.
    Instead, we show counterintuitive results: the weak biased assimilation and the weak in-group bias speed up consensus while the weak Matthew effect plays little role. 
    Compared with the weak in-group bias, the weak biased assimilation speeds up consensus more.
    The simulation data points are calculated by averaging over $500$ independent runs.
    Here $N=100$, $L=200$, $k_0=1$ and $\phi=0.01$ (top panel).
    The lines in the top panel are the theoretical solutions where the slopes are predicted by Eqs.~\eqref{eq: alpha1}-\eqref{eq: beta1} and $\varepsilon=1$.
    The line in the bottom panel is the theoretical solution in Ref.~\cite{WangXunlongChaos}.}
    \label{fig: weak1}
\end{figure}

\section{Evolutionary trajectory}

\begin{figure}[htb]
    \centering
    \includegraphics[width=0.99\linewidth]{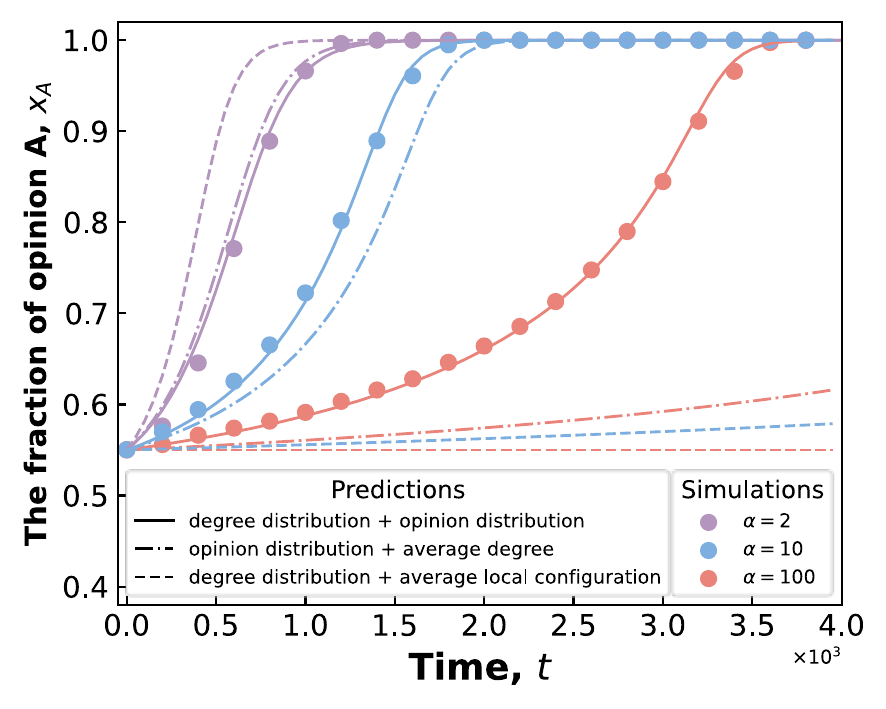}
    \caption{The predictions, taking into account both degree distribution and opinion distribution, i.e., $\pi_{k,m}^X,$ perfectly fit the simulations.
    The predictions without either degree distribution or opinion distribution are inaccurate, which gets worse as the non-linearity ($\alpha$) increases.
    In particular, they lead to wrong predictions when the non-linearity becomes so high ($\alpha\to\infty$) that our model degenerates to the threshold model \cite{PRX_2013_SKM}.
    It implies that both distributions are essential for accuracy.
    Here $N=100$, $L=500$, $k_d=k_s=1$, $\phi=0.01$, $\beta=0$ and the initial fraction of opinion $x_A(t=0)=0.55$.
    The simulation results are averages over $100$ realizations.}
    \label{fig: phase1}
\end{figure}

The exact network configuration $\pi_{k,m}^X$ is also able  to predict the evolutionary trajectory for nonlinear contagion dynamics, i.e., the fraction of opinion $A$ over time.
The evolution of opinion $A$ is typically governed by the mean-field equation, $d{x_A}/dt={T_j^+}-{T_j^-}$, where $T_j^\pm$ are defined in Eqs.~\eqref{eq: noT1+}\eqref{eq: noT1-}  taking fully into account the exact local configuration.
Compared with approximated master equations \cite{PRX_2013_SKM, PNAS_2012_SKM, Chaos_2016_SKM, PRE2010SKM}, the method based on $\pi_{k,m}^X$ dramatically reduces the dimensionality of the system while maintaining high accuracy (see \cref{fig: phase1}).

We show that both the degree distribution and the opinion distribution in the neighborhood with a given size  are indispensable and necessary for accurate predictions for nonlinear contagion dynamics.
We show this point by the following alternatives.

On one hand,
if degree distribution is given without the opinion distribution, then instead we adopt a typical assumption \cite{PRE2008Kimura, Chaos_2016_SKM, Nature2006SimpleRule} that
the fraction of $A-$neighbors around each $X-$individual is the same, which is estimated by the average fraction $p_X$.
This results in another transition probabilities 
\begin{numcases}{}
    T_j^+=x_B\cdot\sum\limits_{k=0}^\infty \frac{\lambda_B^k}{k!}e^{-\lambda_B}\cdot p_B^{1+\alpha},\\
    T_j^-=x_A\cdot\sum\limits_{k=0}^\infty \frac{\lambda_A^k}{k!}e^{-\lambda_A}\cdot(1-p_A)^{1+\alpha},
\end{numcases}
thus resulting in another mean-field equation.
The estimation ignores the fact that individuals with the same neighbor size can have a different number of $A-$neighbors, which results in a large deviation between predictions and simulations (\cref{fig: phase1}).

On the other hand, if opinion distribution is given without degree distribution, then instead
we adopt another typical assumption \cite{PRE2020Simplex,Du2024AsymmetricGP,RogersGross2013ConsensusTimeA} that the neighbor size of each $X-$individual is the same, which is estimated by $\lambda_X$. 
This results in another transition probabilities 
\begin{numcases}{}
    T_j^+=x_B\lambda_B\sum\limits_{m=0}^{\lambda_B}\tbinom{k}{m}p_B^m(1-p_B)^{k-m}(m/\lambda_B)^{1+\alpha},\\
    T_j^-=x_A\lambda_A\sum\limits_{m=0}^{\lambda_A}\tbinom{k}{m}p_A^m(1-p_A)^{k-m}(1-m/\lambda_A)^{1+\alpha},
\end{numcases}
where $\lambda_X$ is estimated by the nearest integer if it's not integer.
This results in another mean-field equation.
The estimation ignores the fact that the dynamical networks lead to changing degree distribution and individuals with same opinion can have different neighbor sizes.
It performs poorly, especially when the non-linearly is strong (\cref{fig: phase1}).

\section {Discussion}

For the proposed nonlinear opinion dynamics,
we give exact analytical solutions of the local network topology $\pi_{k,m}^X$.
It stores all the information about the nearest neighbors around the focal individual.
It's sufficient to predict the dynamics in most contagion models \cite{PRL_2011_SKM, Chaos_2016_SKM, PNAS_2012_SKM} for both the linear and nonlinear spreading rates.
To arrive at the obtained analytical results, we firstly propose a conjecture on the stationary regime of network topology.
Then we show the conjecture is correct thanks to the underlying Markov chain's reversibility.
The exact local network topology $\pi_{k,m}^X$ helps us arrive at the approximated transition probabilities for contagion dynamics on dynamical networks.
The obtained local network configuration explains why mean-field approximations lead to accurate predictions for linear spreading rates \cite{PlosOne2010Bin, CCC2020, Du2022EvolutionaryDO, Du2024AsymmetricGP}.
This is because the Binomial distribution in $\pi_{k,m}^X$ results in the same transition probabilities as that from mean-field approximations.
On the other hand, it provides closed-form analytical predictions with high accuracy for nonlinear spreading rates, which is never done due to the lack of exact results for local configurations. 

We take both absorbing (consensus) time and opinion dynamics over time \cite{PNAS_2012_SKM, PRL_2011_SKM, PRX_2013_SKM} into account. 
The mean-field equations only depend on the difference between transition probabilities.
The consensus time, however, depends not only on the deterministic terms but also on the stochastic terms driven by the internal noises \cite{Gardiner1986HandbookOStochastic, PRL2005VoterHeterogeneous, RogersGross2013ConsensusTimeA, NJP2008FT}.
Thus, it's challenging to predict the consensus time with high-accuracy due to the inefficiency of the mean-field equations.
Thanks to $\pi_{k,m}^X$, we arrive at the exact transition probabilities by the property of Poisson distribution and Binomial distribution.
It facilitates us to find that both weak biased assimilation and weak in-group bias speed up the consensus while weak Matthew effect has little impact on consensus time.
It's counterintuitive to speed up consensus because the biased assimilation intuitively decreases the probability of shift in opinions and the in-group bias intuitively decreases interactions between different opinions.
Furthermore, we find that the speedup effect via weak biased assimilation is stronger than that via weak in-group bias. 
Meanwhile, it implies that the non-linearity of spreading rates ($\alpha<0$) slows down consensus unlike the non-linearity by group interactions, which speeds it up \cite{PRE2020Simplex,PRE2022groupinteraction,SIAM2022Mason}.
Furthermore, we show 
both the degree distribution  and the opinion distribution  are indispensable and necessary for accurate predictions for nonlinear contagion dynamics.
This explains why our framework is necessary and powerful.

The non-linearity in our model arises from the fact that the focal individual takes the other opinion less likely than that in the voter model \cite{Liggett1985Major}. 
In fact, the proposed modified opinion updating can be equivalently interpreted as a threshold model:
the focal individual changes its opinion provided that the fraction of the other opinion in its neighborhood exceeds a threshold following power law distribution \cite{VVV2019interfaceBiased, PNAS2024Complexcontagion}.
This indicates the non-linearity of our model results from the power law of the threshold.  
The nonlinearity of Refs.~\cite{PNAS2023nonlinear, PRL2021Onge}, however, is caused by the heterogeneous infection rates over hyperedges,  although the underlying contagion is still linear. 
In spite of the intrinsic differences in non-linearity sources, our analytical framework can greatly benefit the exact local configurations of high-order dynamical networks, which can be a promising future direction.

To sum up, we obtain exact analytical results for network topology $\pi_{k,m}^X$, leading to accurate predictions for nonlinear contagion dynamics whereas previous methods cannot achieve.

\appendix
\section{The proof of \cref{prop: time}}\label{appendix: proof time}
Ref.~\cite{BOOK2009Arne} proposes
    \begin{numcases}{}
        t_1=\left(1+\sum\limits_{k=1}^{N-1}\prod\limits_{j=1}^k\gamma_j\right)^{-1}\sum\limits_{k=1}^{N-1}\sum\limits_{l=1}^k\frac{1}{T_l^+}\prod\limits_{j=l+1}^k\gamma_j,\label{eq: t1_}\\
        t_j=-t_1\cdot\sum\limits_{k=j}^{N-1}\prod\limits_{m=1}^k\gamma_m+\sum\limits_{k=j}^{N-1}\sum\limits_{l=1}^k\frac{1}{T_l^+}\prod\limits_{m=l+1}^k\gamma_m,\label{eq: tn_}
    \end{numcases}
where $\gamma_j=T_j^-/T_j^+$ is the ratio of the two transition probabilities.
Thus, we obtain
    \begin{align}
    t_{\frac{N}{2}}=&\sum\limits_{k=N/2}^{N-1}\sum\limits_{l=1}^k\frac{1}{T_l^+}\prod\limits_{m=l+1}^k\gamma_m-t_1\cdot\sum\limits_{k=N/2}^{N-1}\prod\limits_{m=1}^k\gamma_m\label{eq: t0}\\
    =&\sum\limits_{k=N/2}^{N-1}\sum\limits_{l=1}^k\frac{1}{T_l^+}\prod\limits_{m=l+1}^k\gamma_m-\underbrace{\left(\frac{\sum\limits_{k=N/2}^{N-1}\prod\limits_{m=1}^k \gamma_m}{1+\sum\limits_{k=1}^{N-1}\prod\limits_{m=1}^k \gamma_m}\right)}_{\Gamma_1=1-\Phi_{N/2}=1/2}\sum\limits_{k=1}^{N-1}\sum\limits_{l=1}^k\frac{1}{T_l^+}\prod\limits_{m=l+1}^k\gamma_m\label{eq: t1}\\
    =&\sum\limits_{k=N/2}^{N-1}\sum\limits_{l=1}^k\frac{1}{T_l^+}\prod\limits_{m=l+1}^k\gamma_m-\frac{1}{2}\sum\limits_{k=1}^{N-1}\sum\limits_{l=1}^k\frac{1}{T_l^+}\prod\limits_{m=l+1}^k\gamma_m\label{eq: t2}\\
    =&\frac{1}{2}\left[\sum\limits_{l=1}^{N/2}\frac{1}{T_l^+}\left(\sum\limits_{k=N/2}^{N-1}\prod\limits_{m=l+1}^k \gamma_m\right)+\sum\limits_{l=N/2+1}^{N-1}\frac{1}{T_l^+}\left(\sum\limits_{k=l}^{N-1}\prod\limits_{m=l+1}^k \gamma_m\right)-\sum\limits_{l=1}^{N/2}\frac{1}{T_l^+}\left(\sum\limits_{k=l}^{N/2}\prod\limits_{m=l+1}^k \gamma_m\right)\right]\label{eq: t3},
\end{align}
where $\Phi_j$ in Eq.~\eqref{eq: t1} represents the probability that $j$ individuals with opinion $A$ take over the whole population, i.e., $\Phi_{j}=\mathbb{P}r[Z(\infty)=N,Z(0)=j]$.
The expression of $\Phi_j$ is given in Ref.~\cite{BOOK2009Arne}.
Since our model is symmetric, $\Phi_{N/2}=1/2$ holds for arbitrary $\alpha,\delta,\beta$.
Thus $\Gamma_1=1-\Phi_{N/2}=1/2$ in Eq.~\eqref{eq: t1} holds.
Eq.~\eqref{eq: t3} is obtained by exchanging the order of the double sums in Eq.~\eqref{eq: t2}. This approach is analogous to that proposed in Appendix A of Ref.~\cite{NJP2008FT}.

\section{Calculation of consensus time $t_{N/2}(0,0,0)$ of neutral model}\label{appendix: tn}
We compute the consensus time $t_{N/2}(0,0,0)$ based on \cref{eq: t0}.
Since $T_j^+=T_j^-=x_A x_B$ in the neutral model, we obtain
\begin{numcases}{}
    \gamma_j=\frac{T_j^-}{T_j^+}=1,\label{eq: gn}\\
    \frac{1}{T_j^+}=N\left(\frac{1}{j}+\frac{1}{N-j}\right).\label{eq: t-1}
\end{numcases}

The sums in Eqs.~(\ref{eq: t0})-(\ref{eq: t3})  can be interpreted as numerical approximations to the integrals \cite{BOOK2009Arne, WangXunlongChaos}, i.e., $\sum\limits_{n=i}^j\dots\approx\int_{i}^j\dots dn$. The error of this approximation is of $\mathcal{O}(N^{-2})$ and is therefore negligible for large $N$. Then substituting Eqs.~(\ref{eq: gn},\ref{eq: t-1}), we obtain
\begin{align}
    t_1=(N-1)\ln(N-1),\label{eq: t1n}
\end{align}
and the consensus time is given by
\begin{align}
    t_{\frac{N}{2}}&=\sum\limits_{k=N/2}^{N-1}\sum\limits_{l=1}^k\frac{1}{T_l^+}\prod\limits_{m=l+1}^k\gamma_m-t_1\cdot\sum\limits_{k=N/2}^{N-1}\prod\limits_{m=1}^k\gamma_m\nonumber\\
    &=N^2\ln 2+\frac{N}{2}\ln N-N\underbrace{\left[N\ln N-(N-1)\ln (N-1)\right]}_{\Gamma_2}.\label{eq: tt}
\end{align}
Here $\Gamma_2=\mathrm{f}_1(N)-\mathrm{f}_1(N-1)$ where $\mathrm{f}_1(x)=x\ln x$. 
Based on Lagrange's Mean Value Theorem, $\Gamma_2=\mathrm{f}_1(N)-\mathrm{f}_1(N-1)=\mathrm{f}_1^{\prime}(\xi)\cdot[N-(N-1)]= \mathrm{f}_1^{\prime}(\xi)$ where $N-1<\xi<N$. 
Because $\mathrm{f}_1^{\prime}(x)=1+\ln{x}$ is a monotonically increasing function, $\Gamma_2=\mathrm{f}_1^{\prime}(\xi)<\mathrm{f}_1^{\prime}(N)=1+\ln N$, thus $\Gamma_2=\mathcal{O}(\ln{N})$.
Since $N\ln{N}\ll N^2$ for large $N$, we thus ignore the terms of $\mathcal{O}(N \ln N)$.
Thus the consensus time $t_{N/2}$ is $N^2\ln{2}$ for large $N$. 

\section{Approximation of $\partial_\alpha T_j^\pm$}\label{appendix: tran_alpha}
Under $\delta=\beta=0$ and $\alpha\to 0$, the transition probability from $j$ to $j+1$ is given by
\begin{align}
    T_j^+=&x_B\cdot\sum\limits_{k=0}^{\infty} \frac{\Bar{k}^k}{k!}e^{-\Bar{k}}\sum\limits_{m=0}^{k}\tbinom{k}{m}x_A^m(1-x_A)^{k-m}\cdot\left(\frac{m}{k}\right)^{1+\alpha}.\label{eq: a_tran_alpha}
\end{align}
We perform a Taylor expansion of $(m/k)^{1+\alpha}$ and intercept the linear terms. 
Then the transition probability is approximated by
\begin{align}
    T_j^+\approx &x_B\cdot\sum\limits_{k=0}^{\infty} \frac{\Bar{k}^k}{k!}e^{-\Bar{k}}\sum\limits_{m=0}^{k}\tbinom{k}{m}x_A^m(1-x_A)^{k-m}\cdot\left[\frac{m}{k}+\alpha\cdot\frac{m}{k}\ln\frac{m}{k}\right]\nonumber\\
    =&x_Ax_B+\alpha\cdot\underbrace{x_B\cdot\sum\limits_{k=0}^{\infty} \frac{\Bar{k}^k}{k!}e^{-\Bar{k}}\sum\limits_{m=0}^{k}\tbinom{k}{m}x_A^m(1-x_A)^{k-m}\cdot\left[\frac{m\ln m-m\ln k}{k}\right]}_{\partial_\alpha T_j^+\big|_{\alpha=0}}.\nonumber
\end{align}
where $\partial_\alpha T_j^+\big|_{\alpha=0}$ can be equivalently written as
\begin{align}
    \partial_\alpha T_j^+\big|_{\alpha=0}=x_B\cdot\sum\limits_{k=0}^{\infty} \frac{\Bar{k}^k}{k!}e^{-\Bar{k}}\left\{\mathbb{E}_{B(k,x_A)}[X\ln{X}]\cdot\frac{1}{k}-\mathbb{E}_{B(k,x_A)}[X]\cdot\frac{\ln{k}}{k}\right\}.\label{eq: pa}
\end{align}
Here we perform a Taylor expansion for $\mathbb{E}_{B(k,x_A)}[X\ln{X}]$, i.e., 
\begin{align}
    \mathbb{E}_{B(k,x_A)}[\mathrm{f}(X)]=&\mathbb{E}_{B(k,x_A)}\left[\mathrm{f}(\mu)+\mathrm{f}^{\prime}(\mu)(X-\mu)+\frac{1}{2}\mathrm{f}^{\prime\prime}(\mu)(X-\mu)^2+o((X-\mu)^2)\right],\nonumber\\
    =&\mathrm{f}(\mu)+\mathrm{f}^{\prime}(\mu)\mathbb{E}_{B(k,x_A)}\left[(X-\mu)\right]+\frac{1}{2}\mathrm{f}^{\prime\prime}(\mu)\mathbb{E}_{B(k,x_A)}\left[(X-\mu)^2\right]+o\left(\mathbb{E}_{B(k,x_A)}\left[(X-\mu)^2\right]\right),\label{eq: log}
\end{align}
where $\mathrm{f}(X)=X\ln{X}$ and $\mu=\mathbb{E}_{B(k,x_A)}[X]=kx_A$. 
We ignore $o\left(\mathbb{E}_{B(k,x_A)}\left[(X-\mu)^2\right]\right)$, thus 
\begin{align}
    \mathbb{E}_{B(k,x_A)}[X\ln{X}]=&\mu\ln{\mu}+(1+\ln{\mu})\mathbb{E}_{B(k,x_A)}[X-\mu]+\frac{1}{2\mu}\mathbb{E}_{B(k,x_A)}\left[(X-\mu)^2\right]\nonumber\\
    =&kx_A\ln{(kx_A)}+\frac{1}{2kx_A}\mathbb{D}[X]\nonumber\\
    =&x_Ak\ln{k}+kx_A\ln{x_A}+\frac{1-x_A}{2},\label{eq: XlogX}
\end{align}
where $\mathbb{D}[X]$ is the variance of the random variable $X$. Substituting Eq.~\eqref{eq: XlogX} into $\mathbb{E}_{B(k,x_A)}[X\ln{X}]$, the Eq. \eqref{eq: pa} becomes 
\begin{align}
    \partial_\alpha T_j^+\big|_{\alpha=0}=&x_B\cdot\sum\limits_{k=0}^{\infty} \frac{\Bar{k}^k}{k!}e^{-\Bar{k}}\left[x_A\ln{k}+x_A\ln{x_A}+\frac{x_B}{2k}-x_A\ln{k}\right]\nonumber\\
    =&x_B\cdot\sum\limits_{k=0}^{\infty} \frac{\Bar{k}^k}{k!}e^{-\Bar{k}}\left[x_A\ln{x_A}+\frac{x_B}{2k}\right]\nonumber\\
    =&x_Bx_A\ln{x_A}\mathbb{E}_{Poisson(\Bar{k})}[1]+\frac{x_B^2}{2}\mathbb{E}_{Poisson(\Bar{k})}\left[\frac{1}{k}\right].\label{eq: pa2}
\end{align}
Based on the series expansion of exponential integral $Ei(x)=P.V.\int_{-\infty}^x\frac{e^t}{t}dt=\ln{x}+\gamma+\sum\limits_{n=1}^\infty \frac{x^n}{n\cdot n!}$ where $\gamma\approx0.577$ is the Euler–Mascheroni constant \cite{BOOK2014MathematicalforPhysicists}, we obtain
\begin{align}
    \mathbb{E}_{poisson(\Bar{k})}\left[\frac{1}{k}\right]=&\sum\limits_{k=1}^\infty \frac{\Bar{k}^k}{k!}e^{-\Bar{k}}\frac{1}{k}\nonumber\\
    =&e^{-\Bar{k}}\sum\limits_{k=1}^\infty \frac{\Bar{k}^k}{k!\cdot k}\nonumber\\
    =&e^{-\Bar{k}}\cdot\left[Ei(\Bar{k})-\ln{\Bar{k}}-\gamma\right].\label{eq: k-1}
\end{align}
Thus Eq.~\eqref{eq: pa2} becomes
\begin{align}
    \partial_\alpha T_j^+\big|_{\alpha=0}=x_Ax_B\ln{x_A}+x_B^2\cdot\underbrace{\frac{e^{-\Bar{k}}}{2}\left[Ei(\Bar{k})-\ln{\Bar{k}}-\gamma\right]}_F,\label{eq: pa1}
\end{align}
where $F$ is a decreasing function of $\Bar{k}$ and is less than $0.172$ when the average degree is not so small (i.e., $\Bar{k}\ge 3.75$).
Similarly, the transition probability from $j$ to $j-1$ is approximated by
\begin{align}
    T_j^-=x_Ax_B+\alpha\cdot\underbrace{\left[ x_Ax_B\ln x_B+x_A^2F\right]}_{\partial_\alpha T_j^-\big|_{\alpha=0}}.\label{eq: pa-}
\end{align}

\section{Calculation of $\partial_\alpha t_{N/2}(0,0,0)$}\label{appendix: t_alpha}
For $\alpha\to 0$ and $\delta=\beta=0$, the transition probabilities are well approximated by $T_j^+=x_Ax_B+\alpha\cdot[x_Ax_B\ln{x_A}+x_B^2F]$ and $T_j^-=x_Ax_B+\alpha\cdot[x_Ax_B\ln{x_B}+x_A^2F]$ when $\varepsilon<j<N-\varepsilon$.
We perform a linear approximation of $\gamma_j=T_j^-/T_j^+$ at $\alpha=0$, i.e.,
\begin{align}
    \gamma_j=&1+\alpha\cdot\underbrace{\frac{\partial_\alpha T_j^--\partial_\alpha T_j^+}{x_Ax_B}\bigg|_{\alpha=0}}_{\partial_\alpha \gamma_j\big|_{\alpha=0}}+o(\alpha)\label{eq: ga}\\
    \approx& 1+\alpha\cdot\left[F\left(\frac{1}{x_B}-\frac{1}{x_A}\right)+\ln\frac{x_B}{x_A}\right].\label{eq: ga1}
\end{align}

As shown in \cref{fig: tg}(b), $\partial_\alpha T_j^+\big|_{\alpha=0}$ and $\partial_\alpha \gamma_j\big|_{\alpha=0}$ are respectively well approximated by Eq. \eqref{eq: pa1} and Eq. \eqref{eq: ga1} but not for $j<\varepsilon$ or $j>N-\varepsilon$, where $\varepsilon=\mathcal{O}(1)\ll N$. In other words, the transition probabilities of $\{Z(t)\}$ in $\{\varepsilon+1,\varepsilon+2,\dots,N-\varepsilon-2,N-\varepsilon-1\}$ can be well approximated.
Thanks to $\varepsilon=\mathcal{O}(1)\ll N$, we calculate $\partial_\alpha t_{N/2}^{\varepsilon\vee N-\varepsilon}$ to approximate $\partial_\alpha t_{N/2}$.
Motivated by the expression of the consensus time Eq. \eqref{eq: t3}, we obtain the near-consensus time $t_{\frac{N}{2}}^{\varepsilon\vee N-\varepsilon}$, which is given by
\begin{align}
    \frac{1}{2}\Bigg[\underbrace{\sum\limits_{l=\varepsilon+1}^{N/2}\frac{1}{T_l^+}\left(\sum\limits_{k=N/2}^{M+\varepsilon-1}\prod\limits_{m=l+1}^k \gamma_m\right)}_{A_1}\!+\!\underbrace{\sum\limits_{l=N/2+1}^{N-\varepsilon-1}\frac{1}{T_l^+}\left(\sum\limits_{k=l}^{M+\varepsilon-1}\prod\limits_{m=l+1}^k \gamma_m\right)}_{A_2}\!-\!\underbrace{\sum\limits_{l=\varepsilon+1}^{N/2}\frac{1}{T_l^+}\left(\sum\limits_{k=l}^{N/2}\prod\limits_{m=l+1}^k \gamma_m\right)}_{A_3}\Bigg],\label{eq: near}
\end{align}
where $M=N-2\varepsilon$. 
It consists of three parts $A_1$, $A_2$ and $A_3$. Then we compute the three parts.

For $A_1$,
\begin{align}
    &\sum\limits_{k=N/2}^{M+\varepsilon-1}\prod\limits_{m=l+1}^k \gamma_m\nonumber\\
    =&\frac{M}{2}-1+\alpha\cdot\Bigg\{\left(\frac{M}{2}-1\right)[(NF+N-l-1)\ln(N-l-1)+(NF+l+1)\ln(l+1)]\nonumber\\
    &-NF\left[(M+\varepsilon-1)\ln(M+\varepsilon-1)-(\varepsilon+1)\ln(\varepsilon+1)+2-M\right]\nonumber\\
    &-\left[\frac{(M+\varepsilon-1)^2}{2}\ln(M+\varepsilon-1)-\frac{(\varepsilon+1)^2}{2}\ln(\varepsilon+1)-\frac{N}{4}(M-2)\right]\Bigg\}.\nonumber
\end{align}
Then
\begin{align}
    &\frac{1}{T_l^+}\left(\sum\limits_{k=N/2}^{M+\varepsilon-1}\prod\limits_{m=l+1}^k \gamma_m\right)\nonumber\\
    =&\left(\frac{M}{2}-1\right)N\left(\frac{1}{l}+\frac{1}{N-l}\right)+\alpha N \cdot\nonumber\\
    &\Bigg\{\!\left(\frac{M}{2}\!-\! 1\!\right)\!\left[\!(\!NF\!+\!N\!-\!1)\frac{\ln(N\!-\!l\!-\!1)}{l}\!+\!(\!NF\!+\!1\!)\frac{\ln(l+1)}{l}\!+\!(\!NF\!-\!1)\frac{\ln(N\!-\!l\!-\!1)}{N\!-\!l}\!+\!(\!NF\!+\!N\!+\!1\!)\frac{\ln(l\!+\!1)}{N-l}\right]\nonumber\\
    &-NF\left[(M+\varepsilon-1)\ln(M+\varepsilon-1)-(\varepsilon+1)\ln(\varepsilon+1)+2-M\right]\left(\frac{1}{l}+\frac{1}{N-l}\right)\nonumber\\
    &-\left[\frac{(M+\varepsilon-1)^2}{2}\ln(M+\varepsilon-1)-\frac{(\varepsilon+1)^2}{2}\ln(\varepsilon+1)-\frac{N}{4}(M-2)\right]\left(\frac{1}{l}+\frac{1}{N-l}\right)\nonumber\\
    &-\left(\frac{M}{2}-1\right)\left[\ln l\left(\frac{1}{l}+\frac{1}{N-l}\right)-\ln N(\frac{1}{l}+\frac{1}{N-l})+\frac{NF}{l^2}\right]\Bigg\}.\nonumber
\end{align}
And then
\begin{align}
    A_1=&\left(\frac{M}{2}-1\right)N\left[\ln(M+\varepsilon-1)-\ln(\varepsilon+1)\right]+\alpha N\cdot\nonumber\\
    &\Bigg\{B_1-NF\left[(M+\varepsilon-1)\ln(M+\varepsilon-1)-(\varepsilon+1)\ln(\varepsilon+1)+2-M\right]\cdot\left[\ln(N-\varepsilon-1)-\ln(\varepsilon+1)\right]\nonumber\\
    &-\left[\frac{(M+\varepsilon-1)^2}{2}\ln(M+\varepsilon-1)-\frac{(\varepsilon+1)^2}{2}\ln(\varepsilon+1)-\frac{N}{4}(M-2)\right]\cdot\left[\ln(N-\varepsilon-1)-\ln(\varepsilon+1)\right]\nonumber\\
    &+\left(\frac{M}{2}-1\right)\ln N\cdot\left[\ln(N-\varepsilon-1)-\ln(\varepsilon+1)\right]-NF\left(\frac{M}{2}-1\right)\left(\frac{1}{\varepsilon+1}-\frac{2}{N}\right)\nonumber\\
    &-\left(\frac{M}{2}-1\right)\frac{1}{2}\cdot\left[\ln^2\frac{N}{2}-\ln^2(\varepsilon+1)\right]\Bigg\}\nonumber\\
    =& \left(\frac{M}{2}-1\right)N\left[\ln(N-\varepsilon-1)-\ln(\varepsilon+1)\right]+\alpha N\cdot\nonumber\\
    &\Bigg\{\!B_1\!-\!NF(N\!-\!\varepsilon\!-\!1)\ln^2(N\!-\!\varepsilon\!-\!1)\!
    +\!NF(N\!-\!\varepsilon\!-\!1)\ln(\varepsilon\!+\!1)\ln(N\!-\!\varepsilon\!-\!1)\!-\!\frac{(N\!-\!\varepsilon\!-\!1)^2}{2}\ln^2(N\!-\!\varepsilon\!-\!1)\nonumber\\
    &+N^2F\ln(N\!-\!\varepsilon\!-\!1)\!+\!\frac{N}{4}(N\!-\!2\varepsilon\!-\!2)\ln(N\!-\!\varepsilon\!-\!1)
    \!+\!\frac{(N\!-\!\varepsilon\!-\!1)^2}{2}\ln(\varepsilon\!+\!1)\ln(N\!-\!\varepsilon\!-\!1)\!+\!\mathcal{O}(N^2)\Bigg\}\nonumber\\
    =& \left(\frac{M}{2}- 1\right)N\left[\ln(N-\varepsilon- 1)-\ln(\varepsilon+ 1)\right]+\alpha \cdot\nonumber\\
    &\left\{\!N B_1\!-\!N^3F\ln^2 N\!+\!N^3F\ln(\varepsilon\!+\!1)\ln N\!-\!\frac{N^3}{2}\ln^2N\!+\!N^3F\ln N\!+\!\frac{N^3}{4}\ln N\!+\!\frac{N^3}{2}\ln(\varepsilon\!+\!1)\ln N\!+\!\mathcal{O}(N^3)\!\right\}\nonumber\\
    =& \left(\frac{M}{2}- 1\right)N\left[\ln(N-\varepsilon- 1)-\ln(\varepsilon+ 1)\right]+\alpha \cdot\nonumber\\
    &\left\{N B_1-N^3\ln^2 N \left(F+\frac{1}{2}\right)+N^3\ln N\left[F+F\ln(\varepsilon+1)+\frac{1}{4}-\frac{\ln(\varepsilon+1)}{2}\right]+\mathcal{O}(N^3)\right\}\label{eq: A1}
\end{align}
where
\begin{align}
    B_1=&\left(\frac{M}{2}-1\right)\cdot\sum\limits_{l=\varepsilon+1}^{N/2}\bigg[(NF+N-1)\frac{\ln(N-l-1)}{l}+(NF+1)\frac{\ln(l+1)}{l}\nonumber\\
    &+(NF-1)\frac{\ln(N-l-1)}{N-l}+(NF+N+1)\frac{\ln(l+1)}{N-l}-\frac{\ln l}{N-l}\bigg].
\end{align}

Similarly, we obtain 
\begin{align}
    A_2=&N(M\!+\!\varepsilon\!-\!1)\!\left\{\ln(M\!+\!\varepsilon\!-\!1)\!-\!\ln(\varepsilon\!+\!1)\!-\!\left[\ln(\frac{N}{2}\!+\!1)\!-\!\ln(\frac{N}{2}\!-\!1)\!\right]\!\right\}\!+\!N^2\left[\ln(\varepsilon\!+\!1)\!-\!\ln(\frac{N}{2}\!-\!1)\!\right]\!+\!\alpha \cdot\nonumber\\
    &\left\{N B_2+N^3\ln N\bigg(4F-2F\ln 2-\ln 2\right)+\mathcal{O}(N^3)\bigg\},\label{eq: A2}
\end{align}
and 
\begin{align}
    A_3=&\frac{N^2}{2}\left[\ln(N-\varepsilon-1)-\ln(\varepsilon+1)\right]+N^2\left[\ln\frac{N}{2}-\ln(N-\varepsilon-1)\right]+\alpha \cdot\nonumber\\
    &\Bigg\{\!NB_3\!-\!N^3\ln^2N\!\left(\!F\!+\!\frac{1}{2}\!\right)\!+\!N^3\ln N\left[2F\ln 2\!+\!F\ln(\varepsilon\!+\!1)\!+\!F\!+\!\ln 2\!+\!\frac{\ln(\varepsilon\!+\!1)}{2}\!+\!\frac{1}{4}\right]\!+\!\mathcal{O}(N^3)\!\Bigg\},\label{eq: A3}
\end{align}
where
\begin{align}
    B_2=&(N-\varepsilon-1)\sum\limits_{l=N/2+1}^{N-\varepsilon-1}\left[(NF+N-1)\frac{\ln(N-l-1)}{l}+(NF+1)\frac{\ln(l+1)}{l}\right]\nonumber\\
    &-(\varepsilon+1)\sum\limits_{l=N/2+1}^{N-\varepsilon-1}\left[(NF-1)\frac{\ln(N-l-1)}{N-l}+(NF+N+1)\frac{\ln(l+1)}{N-l}\right]+(\varepsilon+1)\sum\limits_{l=N/2+1}^{N-\varepsilon-1}\frac{\ln l}{N-l}.
\end{align}
and
\begin{align}
    B_3=&\frac{N}{2}\sum\limits_{l=\varepsilon+1}^{N/2}\bigg[(NF+N-1)\frac{\ln(N-l-1)}{l}+(NF+1)\frac{\ln(l+1)}{l}\nonumber\\
    &-(NF-1)\frac{\ln(N-l-1)}{N-l}-(NF+N+1)\frac{\ln(l+1)}{N-l}+\frac{\ln l}{N-l}\bigg].
\end{align}

Bring Eqs.~\eqref{eq: A1}, \eqref{eq: A2} and \eqref{eq: A3} into Eq.~\eqref{eq: near}, we obtain for $\alpha=0$,
\begin{align}
    t_{\frac{N}{2}}^{\varepsilon\vee N-\varepsilon}\bigg|_{\alpha=0}=&\frac{1}{2}(A_1+A_2-A_3)\bigg|_{\alpha=0}\nonumber\\
    =&\left(\frac{M}{2}-1\right)\frac{N}{2}\left[\ln(N-\varepsilon-1)-\ln(\varepsilon+1)\right]\nonumber\\
    &+\!\frac{N}{2}(M\!+\!\varepsilon\!-\!1)\left\{\ln(M\!+\!\varepsilon\!-\!1)\!-\!\ln(\varepsilon\!+\!1)\!-\!\left[\ln(\frac{N}{2}\!+\!1)\!-\!\ln(\frac{N}{2}\!-\!1)\!\right]\!\right\}\!+\!\frac{N^2}{2}\!\left[\ln(\varepsilon\!+\!1)\!-\!\ln(\frac{N}{2}\!-\!1)\!\right]\nonumber\\
    &-\left\{\frac{N^2}{4}\left[\ln(N-\varepsilon-1)-\ln(\varepsilon+1)\right]+\frac{N^2}{2}\left[\ln\frac{N}{2}-\ln(N-\varepsilon-1)\right]\right\}\nonumber\\
    \approx& \  \ln 2\cdot N^2+N\big[\underbrace{(N-\varepsilon-1)\ln(N-\varepsilon-1)-N\ln N}_{C_1}+(\varepsilon+1)\ln(\varepsilon+1)\big]\label{eq: near1}\\
    =&\  \ln 2\cdot N^2+\mathcal{O}(N\ln{N})\label{eq: near2}
\end{align}
where $C_1=\mathcal{O}(\ln{N})$ in Eq.~\eqref{eq: near1} based on Lagrange's Mean Value Theorem. 
Thus $t_{\frac{N}{2}}^{\varepsilon\vee N-\varepsilon}(0,0,0)=\ln{2}\cdot N^2$ is equal to $t_{N/2}(0,0,0)$ for large $N$, which again implies that it's correct to use the near-consensus time $t_{\frac{N}{2}}^{\varepsilon\vee N-\varepsilon}$ to estimate the consensus time $t_{\frac{N}{2}}$.

Bring the partial derivatives of the three components $A_1$, $A_2$ and $A_3$ into Eq.~\eqref{eq: near}, we obtain
\begin{align}
    \partial_\alpha t_{\frac{N}{2}}^{\varepsilon\vee N-\varepsilon}(0,0,0)=&\frac{1}{2}\frac{\partial}{\partial\alpha}(A_1+A_2-A_3)\bigg|_{\alpha=0}\nonumber\\
    =&\frac{1}{2}\left\{N(B_1+B_2-B_3)+N^3\ln N\left[4F-\ln(\varepsilon+1)-2\ln 2-4F\ln 2\right]\right\}+\mathcal{O}(N^3).\label{eq: paaa}
\end{align}
Then we compute $B_1+B_2-B_3$, which is given by
\begin{align}
    &B_1+B_2-B_3\nonumber\\
    =&\left(\!\frac{M}{2}\!-\!1\!\right)\!\sum\limits_{l=\varepsilon+1}^{N/2}\left[(NF\!+\!N\!-\!1)\frac{\ln(N\!-\!l\!-\!1)}{l}\!+\!(NF\!+\!1)\frac{\ln(l\!+\!1)}{l}\!
    +\!(NF\!-\!1)\frac{\ln(N\!-\!l\!-\!1)}{N\!-\!l}\!+\!(NF\!+\!N\!+\!1)\frac{\ln(l\!+\!1)}{N\!-\!l}\!\right]\nonumber\\
    &+(N-\varepsilon+1)\sum\limits_{l=N/2+1}^{N-\varepsilon-1}\left[(NF+N-1)\frac{\ln(N-l-1)}{l}+(NF+1)\frac{\ln(l+1)}{l}\right]\nonumber\\
    &-(\varepsilon+1)\sum\limits_{l=N/2+1}^{N-\varepsilon-1}\left[(NF-1)\frac{\ln(N-l-1)}{N-l}+(NF+N+1)\frac{\ln(l+1)}{N-l}\right]-\frac{N}{2}\cdot\nonumber\\
    &\sum\limits_{l=\varepsilon+1}^{N/2}\left[(NF\!+\!N\!-\!1)\frac{\ln(N\!-\!l\!-\!1)}{l}\!+\!(NF\!+\!1)\frac{\ln(l\!+\!1)}{l}
    \!-\!(NF\!-\!1)\frac{\ln(N\!-\!l\!-\!1)}{N\!-\!l}\!-\!(NF\!+\!N\!+\!1)\frac{\ln(l\!+\!1)}{N\!-\!l}\right]\!+\!\mathcal{O}(N^2)\nonumber\\
    =&(\!N\!-\!\varepsilon\!+\!1)\!\left\{\sum\limits_{l=N/2+1}^{N-\varepsilon-1}\left[(NF\!+\!N)\frac{\ln(N\!-\!l\!-\!1)}{l}\!+\!NF\frac{\ln(l\!+\!1)}{l}\right]\!+\!\sum\limits_{l=\varepsilon+1}^{N/2}\left[NF\frac{\ln(N\!-\!l)}{N\!-\!l}\!+\!(NF\!+\!N)\frac{\ln l}{N\!-\!l}\right]\!+\!\mathcal{O}(N^2)\!\right\}\nonumber\\
    =&2(N-\varepsilon+1)NF\underbrace{\sum\limits_{l=N/2+1}^{N-\varepsilon-1}\frac{\ln l}{l}}_{C_2}+2(N-\varepsilon+1)(NF+N)\underbrace{\sum\limits_{l=N/2+1}^{N-\varepsilon-1}\frac{\ln(N-l)}{l}}_{C_3}+\mathcal{O}(N^2),\label{eq: BBB}
\end{align}
where
\begin{align}
    C_2\approx&\left[\int_a^b\frac{\ln x}{x}dx\right]\bigg|_{a=\frac{N}{2}+1,b=N-\varepsilon-1}\nonumber\\
    =&\frac{1}{2}\left[\ln^2(N-\varepsilon-1)-\ln^2\frac{N}{2}\right]\nonumber\\
    =&\frac{1}{2}\left[\ln^2N-\ln^2\frac{N}{2}\right]+\mathcal{O}(\frac{\ln N}{N})\nonumber\\
    =&\ln 2\cdot\ln N+\mathcal{O}(1)
\end{align}
and 
\begin{align}
    C_3\approx&\left[\int_a^b\frac{\ln(N-x)}{x}dx\right]\bigg|_{a=\frac{N}{2}+1,b=N-\varepsilon-1}\nonumber\\
    =&\left[\int_a^b\frac{1}{x}\left(\ln N-\sum\limits_{n=1}^\infty \frac{x^n}{nN^n}\right)dx\right]\bigg|_{a=\frac{N}{2}+1,b=N-\varepsilon-1}\nonumber\\
    =&\left[\int_a^b\left(\frac{\ln N}{x}-\sum\limits_{n=1}^\infty \frac{x^{n-1}}{nN^n}\right)dx\right]\bigg|_{a=\frac{N}{2}+1,b=N-\varepsilon-1}\nonumber\\
    =&\left[\ln N\int_a^b\frac{dx}{x}-\sum\limits_{n=1}^\infty \frac{\int_a^b x^{n-1}dx}{nN^n}\right]\bigg|_{a=\frac{N}{2}+1,b=N-\varepsilon-1}\nonumber\\
    =&\left[\ln N\left(\ln b-\ln a\right)-\sum\limits_{n=1}^\infty\frac{b^n-a^n}{n^2N^n}\right]\bigg|_{a=\frac{N}{2}+1,b=N-\varepsilon-1}\label{eq: C3}.
\end{align}
Based on the dilogarithm function $Li_2(x)=\sum\limits_{n=1}^\infty\frac{x^k}{k^2}$, we obtain
\begin{align}
    C_3=&\left\{\ln N\left(\ln b-\ln a\right)-\left[Li_2\left(\frac{b}{N}\right)-Li_2\left(\frac{a}{N}\right)\right]\right\}\bigg|_{a=\frac{N}{2}+1,b=N-\varepsilon-1}\nonumber\\
    =&\ln N\left[\ln(N-\varepsilon-1)-\ln\left(\frac{N}{2}+1\right)\right]-\left[Li_2\left(\frac{N-\varepsilon-1}{N}\right)-Li_2\left(\frac{1}{2}+\frac{1}{N}\right)\right]\nonumber\\
    =&\ln N\left(\ln N-\ln\frac{N}{2}\right)-\left[Li_2\left(1\right)-Li_2\left(\frac{1}{2}\right)\right]+\mathcal{O}(\frac{\ln N}{N})\nonumber\\
    =&\ln 2\cdot\ln N-\left[\frac{\pi^2}{6}-\left(\frac{\pi^2}{12}-\frac{\ln^22}{2}\right)\right]+\mathcal{O}(\frac{\ln N}{N})\nonumber\\
    =&\ln 2\cdot\ln N+\mathcal{O}(1).
\end{align}
Then 
\begin{align}
    B_1+B_2-B_3=N^2\ln{N}(4F+2)\ln{2}.\label{eq: bbb}
\end{align}
Combined with Eq. \eqref{eq: paaa} and Eq.~\eqref{eq: bbb}, we obtain
\begin{align}
    \partial_\alpha t_{\frac{N}{2}}^{\epsilon\vee N-\varepsilon}(0,0,0)=N^3\ln{N}\cdot\underbrace{\left[2F-\frac{\ln{(\varepsilon+1)}}{2}\right]}_{\Lambda_\alpha}+\mathcal{O}(N^3),\label{eq: Lambda_alpha}
\end{align}
where $\Lambda_\alpha<0$ holds since $\bar{k}\ge 3.75$ and $\varepsilon\ge 1$.
The above calculations use several times the Lagrange's Mean Value Theorem and $\sum\limits_{n=i}^j\dots\approx\int_{i}^j\dots dn$.

\section{Calculation of $\partial_\delta t_{N/2}(0,0,0)$}\label{appendix: t_delta}
We let $\theta=\delta/k_0$ and obtain  $p_B=x_A+\theta\cdot x_Ax_B+o(\theta)$.  Since $\theta\to 0$, the transition probability from $j$ to $j+1$ is approximated by 
\begin{align}
    T_j^+=x_Ax_B+\theta\cdot x_Ax_B^2.\label{eq: a_tran_delta-}
\end{align}
Similarly, the transition probability from $j$ to $j-1$ is approximated by 
\begin{align}
    T_j^-=x_Ax_B+\theta\cdot x_A^2x_B.
\end{align}
Thus
\begin{align}
    \gamma_j=\frac{T_j^-}{T_j^+}=1-\frac{2j-N}{N}\theta.
\end{align}
We compute the consensus time $t_{N/2}$ based on Eq.~\eqref{eq: t2}. Since the sums $\sum\limits_{n=i}^j\dots$ can be approximated by the integral $\int_{i}^j\dots dn$, we obtain
\begin{align}
    \frac{1}{T_l^+}\prod\limits_{j=l+1}^k\gamma_j=\frac{N^2}{l(N-l)}+\theta\underbrace{\left\{\frac{N}{l}+\left[(k-l)-\frac{(k-l)(k+l+1)}{N}\right]\frac{N^2}{l(N-l)}\right\}}_{\partial_\theta \left(\frac{1}{T_l^+}\prod\limits_{j=l+1}^k\gamma_j\right)\bigg|_{\theta=0}}.\nonumber
\end{align}
Thus we obtain 
\begin{align}
    &\partial_\theta\left(\sum\limits_{k=\frac{N}{2}}^{N-1}\sum\limits_{l=1}^k\frac{1}{T_l^+}\prod\limits_{m=l+1}^k\gamma_m\right)\Bigg|_{\theta=0}\nonumber\\
    =&\sum\limits_{k=\frac{N}{2}}^{N-1}\left\{(k+1)(N-k)\left[\ln(N-1)-\ln(N-k)+\ln k\right]-Nk\right\}\nonumber\\
    =&\left[\frac{N(N+1)(N+2)}{12}\ln(N-1)\right]-\left[\left(\frac{N}{3}-\frac{1}{2}\right)\frac{N^2}{4}\ln\frac{N}{2}-\left(\frac{7N}{36}+\frac{1}{4}\right)\frac{N^2}{4}+\frac{N}{4}+\frac{5}{36}\right]-\left[\frac{N^2(3N-2)}{8}\right]\nonumber\\
    &+\left[\!\frac{(N\!-\!1)^3}{6}\ln(N\!-\!1)\!-\!\frac{5}{36}(N\!-\!1)^3\!-\!\left(\!\frac{N}{3}\!+\!\frac{3}{2}\!\right)\frac{N^2}{4}\ln\frac{N}{2}\!+\!\left(\!\frac{7N}{36}\!+\!\frac{7}{4}\!\right)\frac{N^2}{4}\!+\!N(N\!-\!1)\ln(N\!-\!1)\!-\!N(N\!-\!1)\!\right]\nonumber\\
    =&\frac{N^3}{4}\ln(N-1)-\frac{N^3}{6}\ln N+\frac{9N^2-4N-2}{12}\ln(N-1)+\frac{\ln 2}{6}N^3-\frac{5}{12}N^3-\frac{N^2}{2}\ln N+\frac{\ln 2}{2}N^2+\frac{N^2}{6}+\frac{N}{3}\nonumber
\end{align}
and
\begin{align}
    \partial_\theta\left(\sum\limits_{k=1}^{N-1}\sum\limits_{l=1}^k\frac{1}{T_l^+}\prod\limits_{m=l+1}^k\gamma_m\right)\bigg|_{\theta=0}=\frac{N^3+3N^2+2N-6}{6}\ln(N-1)-\frac{N^3}{2}+\frac{N^2}{2}+N.\nonumber
\end{align}
Based on Eq.~\eqref{eq: t2}, the partial derivative of consensus time $t_{N/2}$ with respect to $\delta$ is given by
\begin{align}
    \partial_\delta t_{\frac{N}{2}}(0,0,0)=&\frac{1}{k_0}\left[\partial_\theta\left(\sum\limits_{k=\frac{N}{2}}^{N-1}\sum\limits_{l=1}^k\frac{1}{T_l^+}\prod\limits_{m=l+1}^k\gamma_m\right)-\frac{1}{2}\partial_\theta\left(\sum\limits_{k=1}^{N-1}\sum\limits_{l=1}^k\frac{1}{T_l^+}\prod\limits_{m=l+1}^k\gamma_m\right)\right]\nonumber\\
    =&\frac{1}{k_0}\bigg\{\!\frac{N^3}{6}\left[\ln(N\!-\!1)\!-\!\ln N\right]\!+\!\frac{\ln 2\!-\!1}{6}N^3\!+\!\frac{3N^2\!-\!3N\!+\!2}{6}\ln(N\!-\!1)\!-\!\frac{N^2}{2}\!\ln N\!+\!\left(\!\frac{\ln 2}{2}\!-\!\frac{1}{12}\!\right)N^2\!-\!\frac{N}{6}\!\bigg\}\nonumber\\
    =&N^3\underbrace{\frac{\ln 2-1}{6k_0}}_{\Lambda_\delta}+\mathcal{O}(N^2\ln{N}).\label{eq: Lambda_delta}
\end{align}

\section{The perturbation from $\beta\to 0$}\label{appendix: perturbation_beta}

For $\delta=0$, the average degree of opinion $A$ and that of $B$ are the same.
We denote the degree distribution of an arbitrary individual with opinion $X$ by $\pi^X$, where $X\in\{A, B\}$.
It can be expressed by $\pi^X=Poisson(\Bar{k})+\mathcal{O}(\beta)$, since the weak Matthew effect leads to a slight deviation in the degree distribution.

Based on \cref{remark: binomial}, for $\delta=0$, the opinion distribution in the neighborhood with a given size is the Binomial distribution because it's equal to the stationary distribution of $\{Y(t)|X(t)\equiv k\}$.

With the above two properties,  we are able to address how the weak Matthew effect influences the opinion distribution in the neighborhood of Adam, given he holds opinion $A$ and has a fixed neighbor size. More formally, we are addressing how the weak Matthew effect influences the Binomial distribution.
We only need to focus on Eqs.~(\ref{eq: T3},\ref{eq: T4}) since they are the driving force.
Taking the Matthew effect into account, Eq.~\eqref{eq: T3} becomes
\begin{align}
    P_{0,+1}(i,j)=\frac{i-j}{L}\cdot k_{d}\cdot\frac{1}{2}\cdot \frac{\sum\limits_{p\in G(A)}k_p^\beta}{\sum\limits_{p\in G(A)}k_p^\beta+\sum\limits_{q\in G(B)}k_q^\beta},
\end{align}
where $G(X)$ is the set containing all $X-$individuals and $k_p$ is the degree of individual $p$. 
In contrast to the previous expression Eq.~\eqref{eq: T3}, the only change is the probability of rewiring to an $A-$ individual.
For the weak Matthew effect, it becomes
\begin{align}
    \frac{\sum\limits_{p\in G(A)}k_p^\beta}{\sum\limits_{p\in G(A)}k_p^\beta+\sum\limits_{q\in G(B)}k_q^\beta}=&\frac{Nx_A+\beta\cdot\sum\limits_{p\in G(A)}\ln{k_p}+\mathcal{O}(\beta^2)}{N+\beta\cdot\left(\sum\limits_{p\in G(A)}\ln{k_p}+\sum\limits_{q\in G(B)}\ln{k_q}\right)+\mathcal{O}(\beta^2)}\nonumber\\
    =&x_A+\frac{\beta}{N}\cdot\left[\sum\limits_{p\in G(A)}\ln{k_p}-x_A\sum\limits_{p\in G(A)}\ln{k_p}-x_A\sum\limits_{q\in G(B)}\ln{k_q}\right]+\mathcal{O}(\beta^2)\nonumber\\
    =&x_A+\frac{\beta}{N}\cdot\left[x_B\sum\limits_{p\in G(A)}\ln{k_p}-x_A\sum\limits_{q\in G(B)}\ln{k_q}\right]+\mathcal{O}(\beta^2).\label{eq: x_A_m2}
\end{align}
Here $\sum\limits_{p\in G(A)}\ln{k_p}=Nx_A\mathbb{E}_{\pi^A}\left[\ln{X}\right]=Nx_A\mathbb{E}_{Poisson(\Bar{k})}\left[\ln{X}\right]+\mathcal{O}(\beta)$. 
Thus we obtain
\begin{align}
    \frac{\sum\limits_{p\in G(A)}k_p^\beta}{\sum\limits_{p\in G(A)}k_p^\beta+\sum\limits_{q\in G(B)}k_q^\beta}=&x_A+\beta\cdot x_Ax_B\left\{\mathbb{E}_{Poisson(\Bar{k})}\left[\ln{X}\right]+\mathcal{O}(\beta)-\mathbb{E}_{Poisson(\Bar{k})}\left[\ln{X}\right]-\mathcal{O}(\beta)\right\}+\mathcal{O}(\beta^2)\nonumber\\
    =&x_A+\beta\cdot x_Ax_B\cdot \mathcal{O}(\beta)+\mathcal{O}(\beta^2)\nonumber\\
    =&x_A+\mathcal{O}(\beta^2).\label{eq: rewire_m}
\end{align}
The Eq.~\eqref{eq: rewire_m} is of $\mathcal{O}(\beta^2)$, i.e., the deviations in Eqs.~(\ref{eq: T3}, \ref{eq: T4}) from weak $\beta$ are of $\mathcal{O}(\beta^2)$. 
Thus, the weak Matthew effect isn't enough to destroy the opinion distribution among the neighborhood (i.e., the Binomial distribution).

\section*{Acknowledgments}
We would like to acknowledge the assistance of volunteers in putting
together this example manuscript and supplement.

\bibliographystyle{siamplain}
\bibliography{references}
\end{document}